\documentclass[letterpaper]{article} 
\usepackage{aaai2026}  
\usepackage{times}  
\usepackage{helvet}  
\usepackage{courier}  
\usepackage[hyphens]{url}  
\usepackage{graphicx} 
\usepackage{natbib}  
\usepackage{caption} 
\usepackage{algorithm}
\usepackage{algorithmic}
\newcommand{\correspondingauthor}{*}
\usepackage{newfloat}
\usepackage{listings}
\DeclareCaptionStyle{ruled}{labelfont=normalfont,labelsep=colon,strut=off} 
\floatstyle{ruled}
\newfloat{listing}{tb}{lst}{}
\floatname{listing}{Listing}
\usepackage{subfigure}
\usepackage{subcaption}
\usepackage{multirow}
\usepackage{xcolor} 
\usepackage{amssymb}
\usepackage{booktabs}
\usepackage{amsmath}
\usepackage{amsthm}
\newtheorem{definition}{Definition}[section]
\newtheorem{lemma}{Lemma}[section]
\newtheorem{theorem}{Theorem}[section]
\newtheorem{assumption}{Assumption}[section]

\title{Enhanced Privacy Leakage from Noise-Perturbed Gradients via Gradient-Guided Conditional Diffusion Models}
\author{
    Jiayang Meng\textsuperscript{\rm 1}\equalcontrib, Tao Huang\textsuperscript{\rm 2}\equalcontrib, Hong Chen\textsuperscript{\rm 1}\correspondingauthor, Chen Hou\textsuperscript{\rm 2}\correspondingauthor, Guolong Zheng\textsuperscript{\rm 2}
}
\affiliations{
    \textsuperscript{\rm 1}School of Information, Renmin University of China, Beijing 100872, China\\
    \textsuperscript{\rm 2}School of Computer Science and Big Data, Minjiang University, Fuzhou, Fujian 350108, China\\

    \{jiayangmeng, chong\}@ruc.edu.cn, \{huang-tao, houchen, gzheng\}@mju.edu.cn
    
}

\usepackage{bibentry}

\begin{document}

\maketitle

\begin{abstract}
Federated learning synchronizes models through gradient transmission and aggregation. However, these gradients pose significant privacy risks, as sensitive training data is embedded within them. Existing gradient inversion attacks suffer from significantly degraded reconstruction performance when gradients are perturbed by noise-a common defense mechanism. In this paper, we introduce gradient-guided conditional diffusion models for reconstructing private images from leaked gradients, without prior knowledge of the target data distribution. Our approach leverages the inherent denoising capability of diffusion models to circumvent the partial protection offered by noise perturbation, thereby improving attack performance under such defenses. We further provide a theoretical analysis of the reconstruction error bounds and the convergence properties of the attack loss, characterizing the impact of key factors—such as noise magnitude and attacked model architecture—on reconstruction quality. Extensive experiments demonstrate our attack's superior reconstruction performance with Gaussian noise-perturbed gradients, and confirm our theoretical findings.
\end{abstract}



\section{Introduction}

Federated Learning (FL) is a prominent paradigm for collaborative model training, which aggregates locally computed gradients from multiple clients without sharing raw private data \cite{kairouz2021advances, li2020federated}. However, gradients inherently encode sensitive training data, making them vulnerable to Gradient Inversion Attacks (GIAs) \cite{b5,b10,b11}—where adversaries aim to recover private training data from leaked gradients. This privacy risk is especially critical in sensitive domains such as medical imaging and biometric authentication, where private training data is highly valuable and costly to acquire, making it an attractive target for adversaries.

Existing GIAs based on pixel-space differentiation \cite{b5,b6,ig} struggle with high-resolution images, as iterative differentiation of a randomly initialized image becomes computationally prohibitive at higher resolutions. GIAs based on generative models, such as Generative Adversarial Networks (GANs) and diffusion models, require an additional assumption that the pre-training data distribution of the generative model aligns with that of the target images. Furthermore, GAN-based GIAs \cite{gi,gias,ggl,gifd} often suffer from model collapse and training instability, resulting in blurry or artifact-ridden reconstructions. Diffusion model-based GIAs \cite{b13,b14,b15,b25}, which typically rely on model fine-tuning, inherently struggle to precisely recover image details due to the lack of pixel-level guidance.

To mitigate privacy risks associated with gradient leakage, a common strategy is to perturb gradients with noise (e.g., Gaussian or Laplacian) before transmission \cite{b7,b8,b9,b10,b11}. Empirical studies \cite{b25,b26,b27} confirm that this noise injection effectively mitigates GIAs, significantly degrading GIA performance when reconstructing training data from these noise-perturbed gradients.

\begin{figure}[t]
\centering
\includegraphics[width=0.48\textwidth]{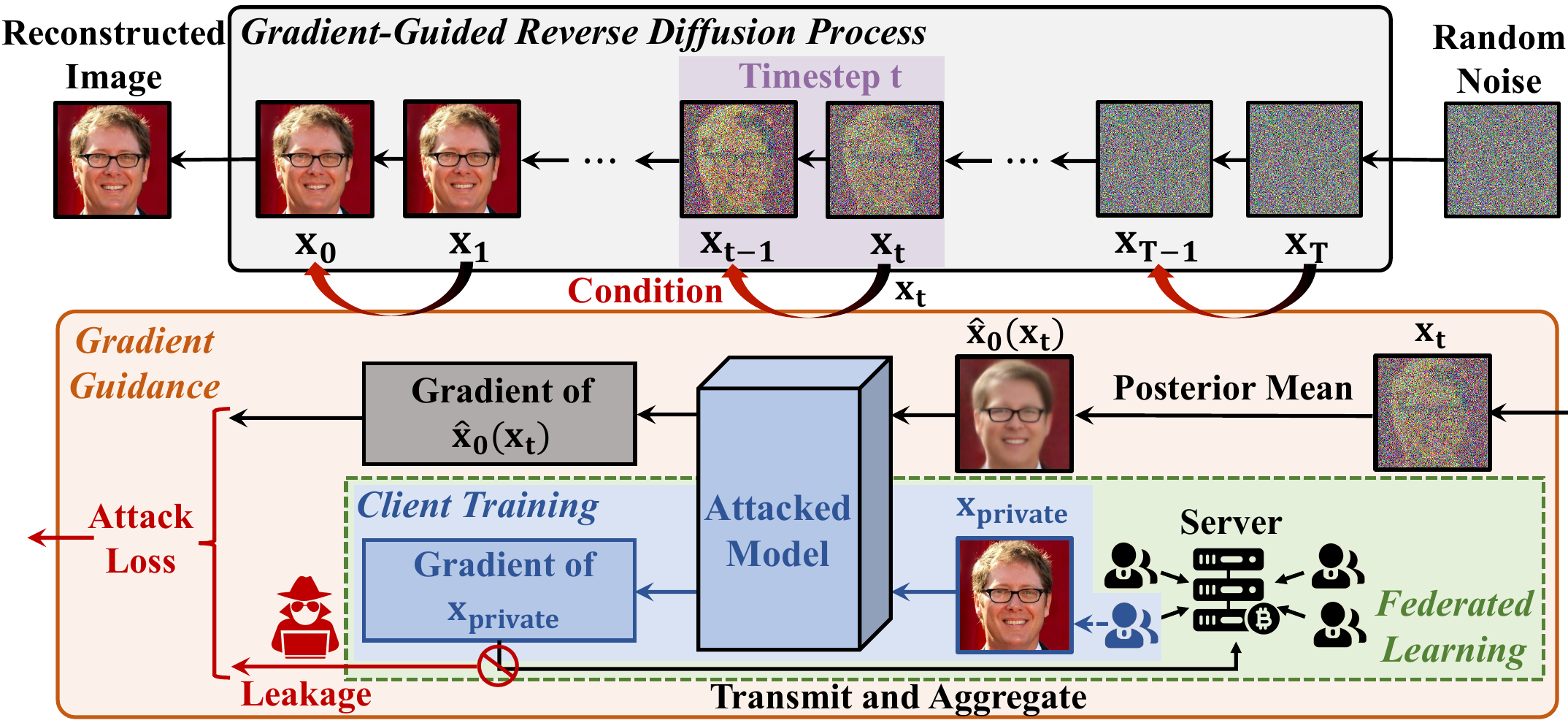}
\caption{An overview of our proposed gradient-guided conditional diffusion model, with a detailed view of the gradient guidance procedure applied at timestep $t$.}
\label{reverse_guide}
\end{figure}

This paper explores the potential of diffusion models' denoising capability for partially circumventing protections based on noise perturbation, thereby improving attack performance. Our contributions are summarized as follows.

\begin{itemize} 
\item We introduce a pixel-level, prior-free GIA method to reconstruct high-fidelity private images from noise-perturbed gradients. Our proposed attack is implemented through Gradient-guided Conditional Diffusion Models (GG-CDMs), as illustrated in Figure \ref{reverse_guide}.
\item Theoretically, we derive the reconstruction error bounds and characterize the convergence properties of the attack loss, quantifying how key factors—such as noise scale and model architecture—affect the reconstruction fidelity. Furthermore, we explicitly propose the RV metric to quantify a model's intrinsic vulnerability to GIAs.
\item We conducted extensive experiments to validate the effectiveness of our proposed GIA method, with results consistently supporting our theoretical findings.
\end{itemize}

\section{Related Work}

\subsection{Pixel-Space Differentiation-Based GIAs}
Gradient inversion attacks (GIAs) based on pixel-space differentiation \cite{b5,b6,ig} typically initialize a differentiable dummy image and iteratively update it by minimizing an attack loss function. Specifically, \cite{b5,b6} minimize the Mean Squared Error (MSE) for pixel-wise accuracy. \cite{ig} improves the reconstruction fidelity by introducing a carefully designed attack loss function. However, these approaches face significant scalability challenges when reconstructing high-resolution images, due to the inherent computational complexity of differentiation. Furthermore, their performance is severely degraded when the leaked gradients are perturbed by noise, as the target image, which corresponds to these noise-perturbed gradients, is distorted.

\subsection{Generative Model-Based GIAs}
Recent advances leverage generative models for high-resolution image reconstruction, avoiding the computational burden of differentiation-based methods. Specifically, \cite{gi,gias,ggl,gifd} have demonstrated that Generative Adversarial Networks (GANs) can serve as effective priors by approximating the natural image manifold, thus improving GIA performance. However, GAN-based methods require distributional alignment between the target image and the GAN’s training set; moreover, they often suffer from training instability. Diffusion models provide a promising alternative, with \cite{b25} demonstrating that model fine-tuning enables stable and noise-free reconstructions. Nevertheless, this fine-tuning paradigm suffers from several limitations: (i) high computational cost due to its iterative nature, (ii) reliance on the distributional prior of target data, (iii) model overfitting that degrades generation quality, and (iv) the absence of pixel-level guidance for fine-grained detail recovery.

\section{Preliminaries}
\subsection{Diffusion Models}

Diffusion models \cite{b29,b31} are a class of generative models capable of producing high-quality and diverse samples. In the forward diffusion process, diffusion models gradually perturb clean data $\mathbf{x}_0 \sim p_{\text{data}}$ by adding Gaussian noise until it becomes pure noise. For both DDPM \cite{ho2020denoising} and DDIM \cite{b30}, the posterior distribution of any $\mathbf{x}_t, t\in[0,T]$ given $\mathbf{x}_0$ is defined as:
\begin{equation}
    \mathbf{x}_t=\sqrt{\alpha_t}\mathbf{x}_0+\sqrt{1-\alpha}_t\epsilon_t, \epsilon_t\sim \mathcal{N}(0,I),
    \label{ddim_ddpm_forward}
\end{equation}
where $\{\alpha_t\}_{t=0}^T$ is the pre-defined noise schedule. The objective is to train a Denoiser, $\epsilon_{\theta}(\mathbf{x}_t,t)$, to predict the noise $\epsilon_t$ added at each timestep $t$. This is achieved by minimizing the expected discrepancy between the true and predicted noise:
\begin{equation}
    L_{\text{Denoiser}} = \mathbb{E}_{t,\mathbf{x}_0,\epsilon_t}[\|\epsilon_t - \epsilon_{\theta}(\mathbf{x}_t,t)\|^2].
    \label{diffusion_Train}
\end{equation}

During the reverse diffusion process, the joint distribution $p(\mathbf{x}_{0:T})$ is given by:
\begin{equation}
    p(\mathbf{x}_{0:T})=p(\mathbf{x}_T) \prod_{t=1}^{T} p_\theta(\mathbf{x}_{t-1} | \mathbf{x}_t),
    \label{reverse_1}
\end{equation}
\begin{equation}
    p_\theta(\mathbf{x}_{t-1} | \mathbf{x}_t) = \mathcal{N}(\mathbf{x}_{t-1}; \mu_\theta(\mathbf{x}_t, t), \Sigma_\theta(\mathbf{x}_t, t)).
    \label{reverse_2}
\end{equation}

Specifically, our attack follows DDIM's reverse sampling:
\begin{equation}
\begin{split}
    \mathbf{x}_{t-1} &= \sqrt{\alpha_{t-1}} 
\underbrace{\left( \frac{\mathbf{x}_t - \sqrt{1 - \alpha_t} \, \epsilon_\theta(\mathbf{x}_t,t)}{\sqrt{\alpha_t}} \right)}_{\text{predicted } \mathbf{x}_0}
\\&+ \underbrace{\sqrt{1 - \alpha_{t-1} - \sigma_t^2 }\cdot \epsilon_\theta(\mathbf{x}_t,t)}_{\text{direction pointing to } \mathbf{x}_t}
+\underbrace{\sigma_t \epsilon_t }_{\text{random noise}},
\end{split}
\end{equation}
where $\sigma_t = \eta \sqrt{\frac{1-\alpha_{t-1}}{1-\alpha_t}} \sqrt{1-\frac{\alpha_t}{\alpha_{t-1}}}$, and $\epsilon_t \sim \mathcal{N}(0,I)$. 

\subsection{Conditional Diffusion Models}

Conditional diffusion models \cite{b16,b17,b18,b22,b23,b24} incorporate a given condition or measurement $\mathbf{y}$ by introducing an additional likelihood term $p(\mathbf{x}_t|\mathbf{y})$. According to the Bayes’ rule,
\begin{equation}
    \nabla_{\mathbf{x}_t} \log p(\mathbf{x}_t|\mathbf{y}) = \nabla_{\mathbf{x}_t} \log p(\mathbf{x}_t) + \nabla_{\mathbf{x}_t} \log p(\mathbf{y} | \mathbf{x}_t).
\end{equation}

With DDIM as the prior, the term $\nabla_{\mathbf{x}_t} \log p(\mathbf{y} | \mathbf{x}_t)$ can be utilized to perform an additional correction step:
\begin{equation}
    \mathbf{x}_{t-1}= \underbrace{\text{DDIM}(\mathbf{x}_t, \epsilon_\theta(\mathbf{x}_t,t))}_{\text{reverse sampling step}} - \underbrace{\gamma\nabla_{\mathbf{x}_t} \log p(\mathbf{y} | \mathbf{x}_t)}_{\text{conditional correction step}}.
    \label{correct}
\end{equation}

\section{Methodology}

\subsection{Problem Statement}



In federated learning, consider a client-side model $F(\mathbf{x}; W)$ parameterized by $W$ under attack. For a private local image $\mathbf{x}_{\text{private}} \in \mathbb{X}$, the corresponding gradient of $F(\mathbf{x}; W)$ is:
\begin{equation}
    \mathbf{g}(\mathbf{x}_{\text{private}}) = \nabla_W F(\mathbf{x}_{\text{private}}; W).
\label{sgd}
\end{equation} 

Under a defense using Gaussian noise perturbation, the leaked gradient $\mathbf{g}_{\text{leaked}} = \mathbf{g}(\mathbf{x}_{\text{private}}) + \mathcal{N}(0, \sigma^2 I) = \mathcal{N}(\mathbf{g}(\mathbf{x}_{\text{private}}), \sigma^2I)$ follows a Gaussian distribution. Given leaked gradient $\mathbf{g}_{\text{leaked}}$, access to the attacked model $F$ (allowing gradient computation for arbitrary inputs), and a diffusion model $\epsilon_{\theta}(\mathbf{x})$ pre-trained on an arbitrary dataset without requirement of distributional alignment with $\mathbf{x}_{\text{private}}$, the objective is to generate an image $\mathbf{x}$ such that $\mathbf{g}(\mathbf{x})$ approximates $\mathbf{g}(\mathbf{x}_{\text{private}})$, thereby ensuring $\mathbf{x}$ approximates $\mathbf{x}_{\text{private}}$.

\subsection{Posterior Distribution Approximation}\label{Approximation}
To achieve computationally efficient and fine-grained pixel-level reconstruction, we guide the DDIM reverse sampling process using leaked gradients as a conditioning signal. This approach enables GIAs within a single reverse diffusion process. A key challenge in this formulation is to derive $p(\mathbf{y} = \mathbf{g}_{\text{leaked}} | \mathbf{x}_t)$ in Equation (\ref{correct}), as $\mathbf{x}_t$ is time-dependent. A tractable approximation can be factorized as:
\begin{equation}
    p(\mathbf{y} = \mathbf{g}_{\text{leaked}} |\mathbf{x}_t) =\int p(\mathbf{y}=\mathbf{g}_{\text{leaked}}|\mathbf{x}_0)p(\mathbf{x}_0|\mathbf{x}_t)d\mathbf{x}_0.
\end{equation}

From DDIM's forward diffusion sampling (Equation~(\ref{ddim_ddpm_forward})), the posterior mean $\hat{\mathbf{x}}_{0}$ conditioned on $\mathbf{x}_t$ can be obtained as:
\begin{equation}
\begin{split}
    \hat{\mathbf{x}}_{0}(\mathbf{x}_t):&=\mathbb{E}[\mathbf{x}_0|\mathbf{x}_t] = \mathbb{E}_{\mathbf{x}_0\sim p(\mathbf{x}_0|\mathbf{x}_t)}[\mathbf{x}_0]\\&=\frac{1}{\sqrt{\alpha_t}} \left( \mathbf{x}_t - (1 - \alpha_t) \epsilon_{\theta}(\mathbf{x}_t,t) \right).
\end{split}
\end{equation}

We then approximate:
\begin{equation}
\begin{split}
    p({\mathbf{g}_{\text{leaked}}} | \mathbf{x}_{t}) \simeq  p({\mathbf{g}_{\text{leaked}}} | \hat{\mathbf{x}}_{0}(\mathbf{x}_t)).
\end{split}
\end{equation}

The error induced by the posterior mean approximation is naturally quantified by the Jensen gap.
\begin{definition}[\textbf{Jensen Gap \cite{pre-trained_diffusion}}]
    Let $\mathbf{x}$ be a random variable with distribution $p(\mathbf{x})$. For some function $f$ that may or may not be convex, Jensen gap is defined as:
    \begin{equation}
        \mathcal{J}(f, \mathbf{x} \sim p(\mathbf{x})) = \mathbb{E}[f(\mathbf{x})] - f(\mathbb{E}[\mathbf{x}]),
    \end{equation}
where the expectation is taken over $p(\mathbf{x})$.
\end{definition}

Specifically, the reconstruction error introduced by the posterior mean approximation is given by $\mathbb{E}_{\mathbf{x}_0|\mathbf{x}_t}[f(\mathbf{x}_0)]-f(\mathbb{E}[\mathbf{x}_0|\mathbf{x}_t])$, where $f$ is a function such as the gradient computation. A smaller Jensen gap implies a lower approximation error and thus higher reconstruction fidelity.

\subsection{Attack Methodology}

In this section, we propose the \textit{Gradient-guided Gaussian Spherical Sampling-Reconstruction (GGSS-R)}, a novel GIA method that leverages Gradient-guided Conditional Diffusion Models (GG-CDMs) to reconstruct private images from leaked gradients via a single reverse diffusion process.

We first introduce a symmetric and differentiable attack loss function $\mathcal{L}\left( \cdot,\cdot \right)=\|\cdot-\cdot\|$. Given the leaked gradient $\mathbf{g}_{\text{leaked}}$, the reverse sampling at timestep $t$ is formulated as:
\begin{equation}
\arg \min_{\mathbf{x}_{t-1}} \left[\nabla_{\mathbf{x}_t}\mathcal{L}\left(\mathbf{g}\left(\hat{\mathbf{x}}_{0}\left(\mathbf{x}_t\right)\right),\mathbf{g}_{\text{leaked}}\right)\right]^T (\mathbf{x}_{t-1} -\mathbf{x}_t),
\end{equation}
where $\mathbf{x}_{t-1} \in CI_{1-\delta}$, with $CI_{1-\delta}$ denoting the $1-\delta$ confidence intervals for the Gaussian distribution defined in Equation (\ref{reverse_2}). This constrained optimization aims to steer the reverse sampling toward directions that most rapidly minimize the attack loss while enforcing sampling within high-confidence intervals under the DDIM Gaussian distribution.

Due to the high-dimensional nature of gradients, direct optimization is challenging. Previous theoretical work \cite{ledoux2001concentration} has established that Gaussian-distributed variables exhibit concentration around their mean. The Laurent-Massart bound \cite{laurent2000adaptive} further provides bounds on the deviation of a Gaussian variable from its mean, as formalized in Lemma \ref{lemma1}.

\begin{lemma}[Laurent-Massart Bound \cite{laurent2000adaptive}]
    For an n-dimensional isotropy Gaussian distribution $\mathbf{x} \sim \mathcal{N}(\mu, \sigma^2I)$, it satisfies:
    \begin{equation}
        P(\|\mathbf{x}-\mu \|^2\geq n\sigma^2+2n\sigma^2(\sqrt{\epsilon}+\epsilon))\leq e^{-n\epsilon},
    \end{equation}
    \begin{equation}
        P(\|\mathbf{x}-\mu \|^2\leq n\sigma^2+2n\sigma^2\sqrt{\epsilon})\leq e^{-n\epsilon},
    \end{equation}
    \label{lemma1}
    where $n$ is $\mathbf{x}$'s dimension.
\end{lemma}

When $n$ is large, $e^{-n\epsilon}$ approaches $0$. In this case, $\mathbf{x}$ is distributed on a sphere with radius $\sqrt{n}\sigma$ centered at $\mu$. The distribution of $\mathbf{x}$ follows: $\mathbf{x} \sim S_{\mu_\theta(\mathbf{x}_t,t),\sqrt{n}\sigma_t}^{n}=\{\mathbf{x}\in \mathbb{R}^n, \|\mathbf{x}-\mu_\theta(\mathbf{x}_t,t)\|=\sqrt{n}\sigma_t\}$. With this spherical constraint, the optimization problem has a closed-form solution:
\begin{equation}
    \mathbf{x}_{t-1}^{*}=\mu_{\theta}(\mathbf{x}_t,t)-\sqrt{n}\sigma_t \frac{\nabla_{\mathbf{x}_t} \mathcal{L}(\mathbf{g}(\hat{\mathbf{x}}_0(\mathbf{x}_t)),\mathbf{g}_{\text{leaked}})}{\| \nabla_{\mathbf{x}_t} \mathcal{L}(\mathbf{g}(\hat{\mathbf{x}}_0(\mathbf{x}_t)),\mathbf{g}_{\text{leaked}})\|}.
    \label{xingxing}
\end{equation}
We term this sampling as \textit{Gradient-guided Gaussian Spherical Sampling (GGSS)}. Figure \ref{dsg_sampling_figure} provides an illustration of GGSS at timestep $t$, and Algorithm \ref{alg:1} outlines the GGSS-R.

\begin{figure}[ht]
\centering
\includegraphics[width=0.31\textwidth]{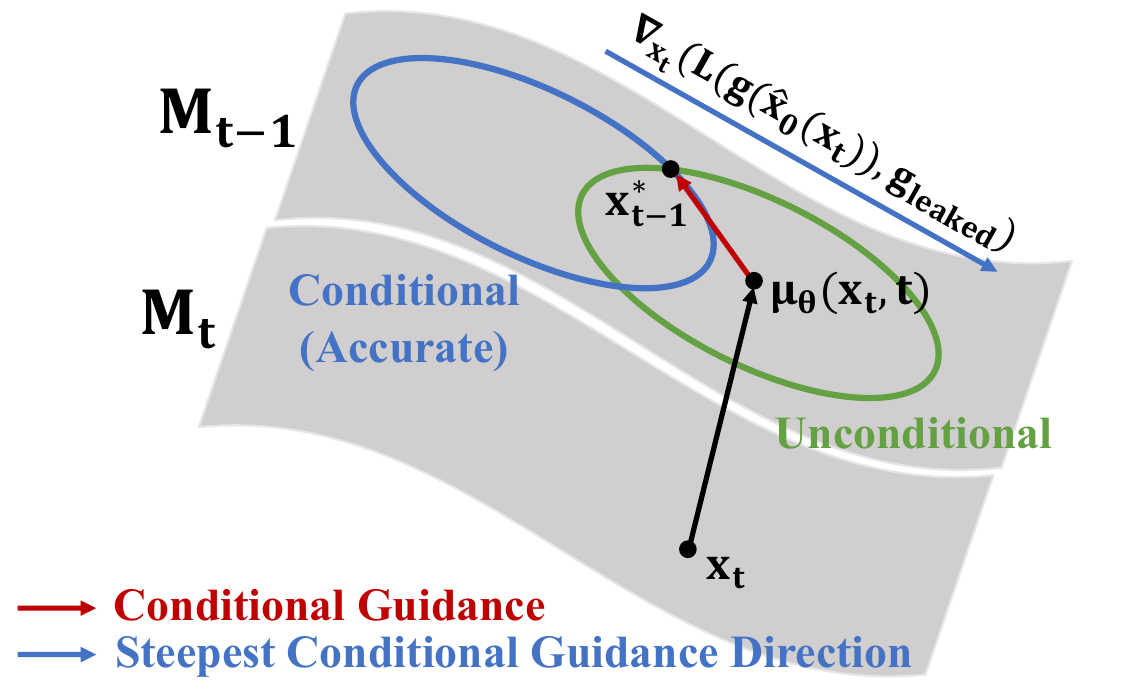}
\caption{An overview of GGSS at timestep $t$, mapping from data manifold $M_t$ to $M_{t-1}$. The green annulus represents the concentration region of samples in unconditional DDIM reverse sampling, and the blue annulus denotes the concentration region under accurate conditional guidance.}
\label{dsg_sampling_figure}
\end{figure}

\begin{algorithm}[ht]
\caption{GGSS-R Process.}
\textbf{Input}: Pre-trained diffusion model $\epsilon_{\theta}$ with maximum timestep $T$; Pre-defined noise schedule $\{\alpha_t\}_{t=1}^T$ and denoising scale $\{\sigma_t\}_{t=1}^T$; Attacked model $F(\mathbf{x}; W)$ where $\mathbf{x}\in \mathbb{R}^n$; Leaked gradient $\mathbf{g}_{\text{leaked}}$; Attack loss function $\mathcal{L}$.\\
\textbf{Output}: Reconstructed image $\mathbf{x}_0$.\\
\vspace{-2.5ex}
\begin{algorithmic}[1] 
\STATE Randomly initialize a dummy image: $\mathbf{x}_T \sim \mathcal{N}(0, I)$.
\FOR{$t=T$ to $1$}
\STATE Predict added noise at forward step $t$: $\hat{\epsilon}_t \leftarrow \epsilon_{\theta}(\mathbf{x}_t, t)$.
\STATE Approximate the posterior mean: \\$\hat{\mathbf{x}}_0(\mathbf{x}_t) \leftarrow \frac{1}{\sqrt{\alpha_t}} \left( \mathbf{x}_t - (1 - \alpha_t) \hat{\epsilon}_t \right)$.
\STATE Compute the mean of $\mathbf{x}_{t-1}$: \\
$\mu_{\theta}(\mathbf{x}_t,t) = \sqrt{\alpha_{t-1}} \hat{\mathbf{x}}_0(\mathbf{x}_t) + \sqrt{1-\alpha_{t-1}-\sigma_t^2} \cdot \hat{\epsilon}_t$.\\
\STATE Input $\hat{\mathbf{x}}_0(\mathbf{x}_t)$ to $F$ to calculate $\mathbf{g}(\hat{\mathbf{x}}_0(\mathbf{x}_t))$.\\
\STATE Perform gradient-guided conditional correction: \\$\mathbf{x}_{t-1} \leftarrow \mu_{\theta}(\mathbf{x}_t,t) -\sqrt{n} \sigma_t \cdot \frac{\nabla_{\mathbf{x}_t} \mathcal{L}(\mathbf{g}(\hat{\mathbf{x}}_0(\mathbf{x}_t)),\mathbf{g}_{\text{leaked}})}{\| \nabla_{\mathbf{x}_t} \mathcal{L}(\mathbf{g}(\hat{\mathbf{x}}_0(\mathbf{x}_t)),\mathbf{g}_{\text{leaked}})\|}$.
\ENDFOR
\STATE \textbf{return} Reconstructed image $\mathbf{x}_0$.
\end{algorithmic}
\label{alg:1}
\end{algorithm}

In practice, deterministic sampling under a Gaussian spherical constraint minimizes the attack loss most rapidly. However, noise-perturbed gradients make the corresponding target image non-smooth and deviate from the original target. To address this, we combine DDIM's unconditional sampling direction with the optimal conditional sampling direction, thereby enhancing both generation quality and diversity. Specifically, line 7 in Algorithm~\ref{alg:1} is refined as:
\begin{equation}
\begin{split}
    \mathbf{x}_{t-1}=\mu_{\theta}&(\mathbf{x}_t,t) + r \frac{d_m}{\|d_m\|}, \\
    d_m = d^{\text{sample}}& + m_r (d^{*} - d^{\text{sample}}),\\
\end{split}
    \label{equation_dsg}
\end{equation}
where $d^{*}=-\sqrt{n} \sigma_t \cdot \frac{\nabla_{\mathbf{x}_t} \mathcal{L}(\mathbf{g}(\hat{\mathbf{x}}_0),\mathbf{g}_{\text{leaked}})}{\| \nabla_{\mathbf{x}_t} \mathcal{L}(\mathbf{g}(\hat{\mathbf{x}}_0),\mathbf{g}_{\text{leaked}})\|}$ represents the optimal adjustment step derived from the leaked gradient $\mathbf{g}_{\text{leaked}}$, while $d^{\text{sample}}=\sigma_t \epsilon_t$ (with $\epsilon_t \sim \mathcal{N}(0,I)$) denotes DDIM's stochastic unconditional sampling step. Their weighted combination, $d_m$, is controlled by the guidance rate $m_r \in[0,1]$, which balances two objectives, and $r$ dictates the overall step size of the conditional guidance.

This refined guidance mechanism strategically balances conditional guidance with sample quality, not only ensuring adherence to gradient-based conditions but also promoting smoother, higher-fidelity image generation. By exploiting DDIM's inherent generative capability, we mitigate the adverse effects of noise-induced perturbations, potentially generating reconstructions closer to the original target.

\section{Theoretical Analysis}

This section establishes theoretical upper and lower bounds on the reconstruction error arising from the posterior mean approximation, which is quantified by the Jensen gap in Section \ref{Approximation}. In addition, we analyze the convergence behavior of the attack loss. Furthermore, our analysis reveals the significant variation in vulnerability to GIAs across model architectures and proposes the $\text{RV}$ metric to quantify a model's intrinsic vulnerability to such attacks. Detailed proofs of Theorems are provided in Appendix A.

\subsection{Bounds of the Reconstruction Error}

For Algorithm \ref{alg:1}, Theorem \ref{UpperBoundJensen} establishes the upper bound of the reconstruction error.

\begin{theorem}[Upper Bound of the Reconstruction Error]
    Under the attacked model $F(\mathbf{x}; W)$, we assume that the adversary obtains a Gaussian noise–perturbed gradient $\mathbf{g}_{\text{leaked}} = \nabla_{W} F(\mathbf{x}_\text{private};W) + \mathcal{N}(0, \sigma^2 I)$. The gradient inversion attack follows Algorithm~\ref{alg:1}, using the approximation $p({\mathbf{g}_{\text{leaked}}} | \mathbf{x}_{t}) \simeq p({\mathbf{g}_{\text{leaked}}} | \hat{\mathbf{x}}_{0}(\mathbf{x}_{t}))$, where $\hat{\mathbf{x}}_{0}(\mathbf{x}_{t})$ denotes the posterior mean of $\mathbf{x}_0$ given $\mathbf{x}_t$. The reconstruction error—measured by the Jensen gap—is upper bounded by:
        \begin{equation*}
        \begin{split}
            &\mathcal{J}\left( \mathbf{g}(\mathbf{x}), p(\mathbf{x}_0 | \mathbf{x}_{t}) \right)=|\mathbb{E}[\mathbf{g}(\mathbf{x}_0)]-\mathbf{g}(\mathbb{E}[\mathbf{x}_0])| \\
            & \leq \frac{n}{\sqrt{2\pi \sigma^{2}}} \| \nabla_{\mathbf{x}} \mathbf{g}(\mathbf{x}) \| \int \| \mathbf{x}_0 - \hat{\mathbf{x}}_{0}(\mathbf{x}_{t}) \|p(\mathbf{x}_0 | \mathbf{x}_{t}) d\mathbf{x}_0,
        \end{split}
        \end{equation*}
        where $n$ is $\mathbf{x}$'s dimension, and $\mathbf{g}(\mathbf{x})=\nabla_{\mathbf{x}} \nabla_{W} F(\mathbf{x}; W)$.
    \label{UpperBoundJensen}
\end{theorem}

The upper bound of the reconstruction error characterizes the worst-case reconstruction quality. The term $\int \| \mathbf{x}_0 - \hat{\mathbf{x}}_{0}(\mathbf{x}_{t}) \|p(\mathbf{x}_0 | \mathbf{x}_{t}) d\mathbf{x}_0$, which depends on the pre-trained diffusion model, is sensitive to the noise magnitude $\sigma^2$. Our analysis further identifies an architecture-dependent factor, $\| \nabla_{\mathbf{x}} \mathbf{g}(\mathbf{x}) \|=\| \nabla_{\mathbf{x}} \nabla_{W} F(\mathbf{x}; W) \|$. To quantify the vulnerability of a model $F$ to GIAs, we introduce the Reconstruction Vulnerability (RV) metric, defined as the Frobenius norm of the upper bound on this quantity as follows.

\begin{definition}[\textbf{Reconstruction Vulnerability (RV)}]
    Let $F(\mathbf{x};W)$ be a model parameterized by $W$, where $\mathbf{x}$ denotes the input. The Reconstruction Vulnerability ($\text{RV}$) of $F$ with respect to a private dataset $\mathbb{X}$ is formally defined as:
    \begin{equation}
        \text{RV} = \max_{\mathbf{x} \in \mathbb{X}} \| \nabla_{\mathbf{x}} \nabla_{W} F(\mathbf{x}; W) \|_F,
    \end{equation}
    where $\|\cdot \|_F$ denotes the Frobenius norm.
\label{defination of rv}
\end{definition}


A larger RV value indicates a greater sensitivity of the parameter gradients to input variations, which amplifies input-specific leakage and facilitates higher-fidelity GIAs. As directly computing the RV is computationally intensive, we project $\nabla_{W} F(\mathbf{z}; W)$ onto $M$ random orthogonal directions $\{\mathbf{v}_j\}_{j=1}^M$, where $\mathbf{v}_j \sim \mathbf{p}_\mathbf{v}$. Given a model $F$ trained on a dataset $\mathbb{X}=\{\mathbf{x}_i\}_{i=1}^{|\mathbb{X}|}, \mathbf{x}_i\sim \mathbf{p}_{\text{data}}$, RV can be estimated as:
\begin{equation}
\begin{split}
    \text{RV} &= \mathbb{E}_{\mathbf{v} \sim \mathbf{p}_\mathbf{v}} \mathbb{E}_{\mathbf{x} \sim \mathbf{p}_{\text{data}}(\mathbf{x})} \| \nabla_{\mathbf{x}} \left( \mathbf{v}^{T} \nabla_{W} F(\mathbf{x};W) \right)\| \\
    & \simeq \frac{1}{N} \frac{1}{M} \sum_{i=1}^{N} \sum_{j=1}^{M} \| \nabla_{\mathbf{x}} \left( \mathbf{v}_{j}^{T} \nabla_{W} F(\mathbf{x}_i;W) \right)\|
\end{split}
\end{equation}

We then examine the lower bound of the reconstruction error. As shown in Equation (\ref{xingxing}), the properties of the attack loss $\mathcal{L}$ influence the lower bound of the Jensen gap. Let $\mathcal{L}\left( {\mathbf{x}_t} \right)=\mathcal{L}\left( \mathbf{g}({\mathbf{x}_t}) \right)=\mathcal{L}\left( \mathbf{g}(\hat{\textbf{x}}_0(\mathbf{x}_t)), \mathbf{g}_{\text{leaked}} \right)$ for simplicity. We rely on Assumptions \ref{assu1}-\ref{assu4} to derive the lower bound.

\begin{assumption}
    $\mathcal{L}(\mathbf{g}(\mathbf{x}))$ is $\alpha$-smooth and $\beta$-convex with respect to $\mathbf{g}(\mathbf{x})$;
    \label{assu1}
\end{assumption}

\begin{assumption}
    $\mathbf{g}(\mathbf{x})$ is $L_g$-smooth with respect to $\mathbf{x}$;
    \label{assu3}
\end{assumption}

\begin{assumption}
    Jacobian $\mathbf{J}_\mathbf{g}(\mathbf{x}) = \nabla_{\mathbf{x}} \mathbf{g}(\mathbf{x})$ is full rank.
    \label{assu4}
\end{assumption}

\begin{theorem}[Lower Bound of the Reconstruction Error]
    Under conditions of Theorem \ref{UpperBoundJensen} and Assumptions \ref{assu1}-\ref{assu4}, for $\mathbf{x} \in \mathbb{R}^n$, the reconstruction error is lower bounded by:
    {
    \begin{equation*}
        \begin{split}
        & \mathcal{J}\left( \mathcal{L}\left( \mathbf{g}(\mathbf{x}),\mathbf{g}_{\text{leaked}} \right), p(\mathbf{x}_0 | \mathbf{x}_{t}) \right)\\
        &=|\mathbb{E}[\mathcal{L}(\mathbf{g}(\mathbf{x}_0),\mathbf{g}_{\text{leaked}})]-\mathcal{L}(\mathbf{g}(\mathbb{E}[\mathbf{x}_0]),\mathbf{g}_{\text{leaked}})| \\
        & \geq \frac{1}{2} \left( \beta \lambda_{min}\left( \mathbf{J}_\mathbf{g}(\mathbf{x})^T\mathbf{J}_\mathbf{g}(\mathbf{x}) \right) - \alpha L_g \right) \sum_{i=1}^{n} \sigma_{i}^{2},
        \end{split}
    \end{equation*}
    }
    where $\lambda_{\text{min}}(\cdot)$ represents the smallest eigenvalue of $\cdot$, $L_g$ can be refined to a specific value $\frac{n}{\sqrt{2 \pi \sigma^2}} \exp(-\frac{1}{2\sigma^2})$ under Gaussian noise-perturbed isotropic gradients, and $\{ \sigma_i^2 \}_{i=1}^n$ are the eigenvalues of the covariance matrix of $p(\mathbf{x}_0 | \mathbf{x}_{t})$.
\label{lowerboundDSG}
\end{theorem}

In Theorem \ref{lowerboundDSG}, $\lambda_{\text{min}}\left( \mathbf{J}_\mathbf{g}(\mathbf{x})^T\mathbf{J}_\mathbf{g}(\mathbf{x}) \right)$ can be interpreted as the information loss associated with the mapping $\mathbf{x} \mapsto \mathbf{g}$. This suggests that reconstructing an exact replica of the original private image from leaked gradients—even without explicit perturbation noise—is impossible. The degree of information loss is scaled by the randomness inherent in the DDIM's denoising steps, specifically by $\sum_{i=1}^{n} \sigma_{i}^{2}$. Moreover, Theorem \ref{lowerboundDSG} demonstrates that a larger Gaussian noise variance $\sigma^2$ monotonically decreases $L_g$, thus increasing the lower bound on the Jensen gap. This indicates a positive correlation between the reconstruction error and the noise scale.

\begin{figure*}[ht]
  \centering
  \subfigure[CelebA]{
		\includegraphics[width=0.32\textwidth]{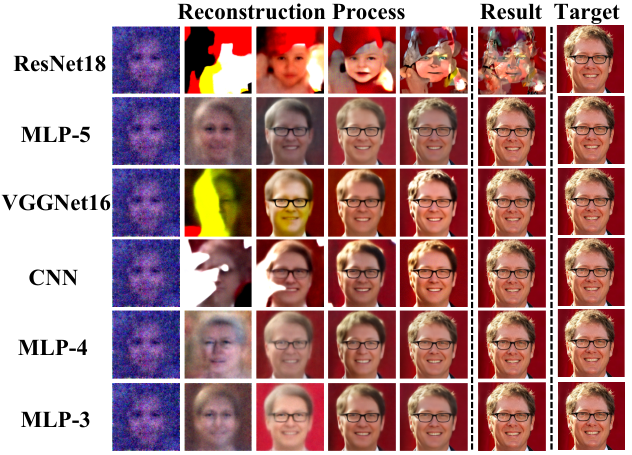}
    }
  \subfigure[LSUN]{
		\includegraphics[width=0.32\textwidth]{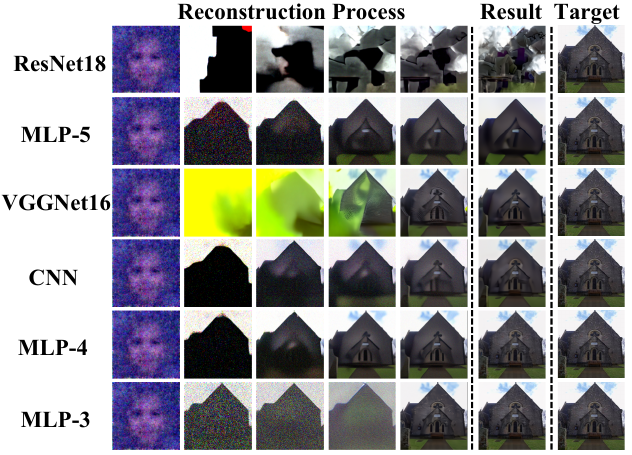}
    }
  \subfigure[ImageNet]{
		\includegraphics[width=0.32\textwidth]{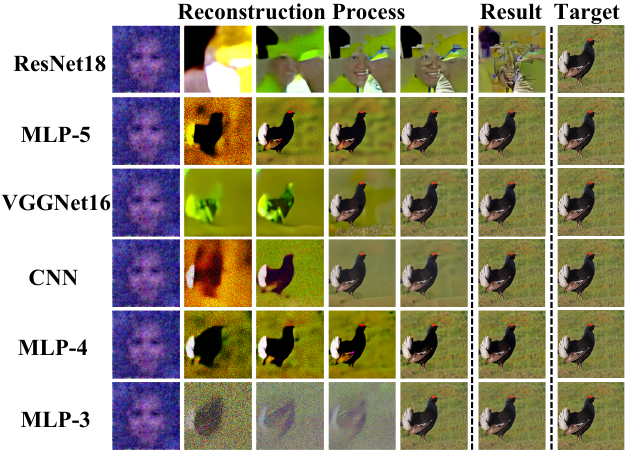}
    }
  \caption{Reconstruction process of GGSS-R across diverse datasets and attacked models. The \textit{Reconstruction Process} columns display the intermediate results $\hat{\mathbf{x}}_0(\mathbf{x}_t)$, with final reconstructions and target image presented in the \textit{Result} and \textit{Target} columns.}
  \label{dsg_more_dataset}
\end{figure*}

\subsection{Convergence of the Attack Loss}

In the context of GIA, a critical issue is whether the attack loss decreases as the reconstructed samples are generated progressively. Formally, we examine whether the sequence $\{ \mathcal{L}(\mathbf{x}_t) \}_{t=T}^{1}$ is monotonically decreasing. If this holds, it is essential to characterize the rate of this decrease. Specifically, Theorem \ref{upperDSG} states the monotonic decrease of the attack loss, providing a lower bound on its decrease rate, and Theorem \ref{DSGlower} derives an upper bound on this rate.

\begin{theorem}[Convergence of the Attack Loss]
    Assume that the attack loss $\mathcal{L}(\mathbf{x})$ is $\alpha^{\prime}$-smooth, $\beta^{\prime}$-strongly convex with respect to $\mathbf{x}$. When GIA is executed following Algorithm~\ref{alg:1}, the sequence $\{ \mathcal{L}(\mathbf{x}_{t}) \}_{t=T}^{1}$ monotonically decreases as the denoising steps progress from timestep $t=T$ to $1$.
    \label{upperDSG}
\end{theorem}

\begin{theorem}[Upper Bound on the Convergence Rate of the Attack Loss]
    Under Assumptions \ref{assu1}-\ref{assu4}, and assume that the attack loss $\mathcal{L}(\mathbf{x})$ is $\alpha^{\prime}$-smooth, $\beta^{\prime}$-strongly convex with respect to $\mathbf{x}$. The optimization using Algorithm \ref{alg:1} satisfies:
        \begin{equation*}
        \begin{split}
            \mathcal{L}(\mathbf{x}_{t-1}) - \mathcal{L}(\mathbf{x}_{t})& \geq  \frac{\beta^{\prime} n \sigma_t^2}{2} +\beta^{\prime} k \sqrt{n} \sigma_t + \frac{\beta^{\prime} k^2}{2}\\
            -  ( k  + &  \sqrt{n} \sigma_t) (\sqrt{\lambda_{\text{max}}(\mathbf{J}_\mathbf{g}(\mathbf{x})\mathbf{J}_\mathbf{g}(\mathbf{x})^T)} \cdot L_g),
        \end{split}
        \end{equation*}
    where $n$ is $\mathbf{x}$'s dimension, $k$ is a positive constant, $\lambda_{\text{max}}(\cdot)$ is the largest eigenvalue of $\cdot$, and $L_g = \frac{n}{\sqrt{2 \pi \sigma^2}} \exp(-\frac{1}{2\sigma^2})$.
    \label{DSGlower}
\end{theorem}

As discussed in Theorem~\ref{lowerboundDSG}, the private information loss associated with the mapping $\mathbf{x} \mapsto \mathbf{g}$ is quantified by the term $\sqrt{\lambda_{\text{max}}(\mathbf{J}_\mathbf{g}(\mathbf{x})\mathbf{J}_\mathbf{g}(\mathbf{x})^T)}$. The convergence rate of the attack loss is further influenced by $L_g$, with $L_g \leq \frac{n}{\sqrt{2 \pi \sigma^2}}$. Consequently, increasing the noise variance $\sigma^2$ tightens the lower bound on $\mathcal{L}(\mathbf{x}_{t-1}) - \mathcal{L}(\mathbf{x}_{t})$, which, in turn, slows the convergence of attack loss and deteriorates reconstruction quality.

\section{Experiments}
\subsection{Experimental Setup}

\subsubsection{Datasets and Models.} Experiments were conducted on three datasets: CelebA \cite{celeba}, LSUN \cite{yu2015lsun}, and ImageNet \cite{deng2009imagenet}, with images of $256 \times 256$ resolution. The attacked models include MLPs with 3/4/5 layers (MLP-3/4/5), CNN, ResNet18 \cite{he2016deep}, and VGGNet16 \cite{simonyan2014very}. The pre-trained diffusion model, sourced from \cite{pre-trained_diffusion}, was pre-trained on the FFHQ face dataset \cite{karras2017progressive}, processing $256 \times 256$ resolution images.

\subsubsection{Hyperparameters.} We employ conditional reverse sampling in Equation (\ref{equation_dsg}), a common practice. We set a gradient acquisition batch size of $1$ and a guidance rate $m_r$ of $0.20$. The attack loss is measured using the Euclidean distance. All $\text{RV}$ values are computed with $M=1000$ and $N = 310$.

\subsubsection{Evaluation Metrics.} Reconstruction quality is evaluated using three widely adopted computer vision metrics: MSE, PSNR, and LPIPS, which quantify the similarity between the reconstructed samples and the original target images.

\subsubsection{Baselines.} We compared our GGSS-R with state-of-the-art baselines, including DLG \cite{b5}, IG \cite{b7}, GI \cite{b8}, GIAS \cite{b10} , GGL \cite{b9}, GIFD \cite{b11}, and FinetunedDiff \cite{b25}.

More detailed experimental settings are provided in Appendix B.

\subsection{Experimental Results}

\subsubsection{GIA from Original Gradients and RV Metric Validation.}

This section presents reconstruction results obtained with original, unperturbed gradients. For each attacked model, we first compute its RV value and then perform our GGSS-R attack. Figure \ref{dsg_more_dataset} visually illustrates the reconstruction results, with attacked models ordered by increasing RV values. Table \ref{RV_table} further quantifies the reconstruction performance.

\begin{table}[ht]
    \centering
    \resizebox{0.48\textwidth}{!}{
    \begin{tabular}{llcccc}
    \hline
    \toprule
       Dataset & Model& RV & MSE$\downarrow$& PSNR$\uparrow$ & LPIPS$\downarrow$\\
    \midrule
             \multirow{6}{*}{CelebA} & ResNet18 & 0.17 & 9.46e-2  & 10.79 & 4.14e-4\\
            & MLP-5 & 8.71 & 3.13e-5  & 45.05 & 9.84e-8 \\
            & VGGNet16 & 10.74 & 1.26e-4 & 40.76 & 2.43e-7 \\
            & CNN & 18.51  & 7.72e-5 & 41.12 & 2.05e-7\\
            & MLP-4 & 20.85 & 3.08e-5 & 45.12 & 8.92e-8\\ 
            & MLP-3 & 52.73 & 4.55e-5  & 43.42 & 1.60e-7\\
            \cline{1-6}
              \multirow{6}{*}{LSUN} & ResNet18 & 0.17 & 1.00e-1  & 9.99 & 3.10e-4\\
            & MLP-5 & 8.71 & 2.27e-3  & 26.44 & 3.09e-5 \\
            & VGGNet16 & 10.74 & 6.63e-3 & 22.79 & 4.88e-5 \\
            & CNN & 18.51  & 2.54e-3 & 25.95 & 2.44e-5\\
            & MLP-4 & 20.85 & 3.61e-4 & 34.43 & 5.60e-7\\ 
            & MLP-3 & 52.73 & 1.48e-4& 38.29 & 4.03e-7\\ 
            \cline{1-6}
             \multirow{6}{*}{ImageNet} & ResNet18 & 0.17 & 4.45e-2   & 13.52 & 2.54e-4\\
            & MLP-5 & 8.71 & 3.34e-4 & 25.35 & 3.58e-5 \\
            & VGGNet16 & 10.74 & 1.07e-3 & 29.71 & 1.48e-6 \\
            & CNN & 18.51  & 1.15e-4 & 27.39 & 2.55e-5\\
            & MLP-4 & 20.85 & 3.09e-4 & 25.10& 4.05e-5\\ 
            & MLP-3 & 52.73 & 1.05e-4 & 28.00 & 1.63e-5\\    
    \bottomrule
    \end{tabular}}
    \caption{Quantitative analysis of the reconstruction quality across various datasets and attacked models.}
    \label{RV_table}
\end{table}

Based on these results, we highlight two key findings: (i) Robustness across diverse distributions. Despite diffusion models being pre-trained on the FFHQ face dataset, GGSS-R effectively reconstructs high-fidelity images across diverse distributions. This demonstrates our attack's robustness, enabling efficient GIAs without requiring prior knowledge of the target data distribution—a common limitation in existing generative model-based approaches; (ii) RV as an effective indicator of a model's vulnerability to GIAs. Models with higher RV values consistently exhibit higher reconstruction quality, indicating an elevated privacy risk. Specifically, we observe a convex increase in reconstruction performance with RV. This finding suggests designing models with lower RV values, as they inherently present lower gradient privacy risks. Furthermore, architectural modifications to low-RV models warrant particular caution, as even minor increases in RV can significantly increase privacy risks.

\subsubsection{GIA from Gaussian Noise-Perturbed Gradients.}

When the leaked gradient is perturbed by Gaussian noise, the target image corresponding to $\mathbf{g}_{\text{leaked}}$ deviates from the original target. We define this distorted target as $\mathbf{x}_{\text{noisy}}$. This section investigates the effectiveness of the added noise in impeding our attack and validates our theoretical findings.

Figure \ref{noisy_celeba_process} illustrates the reconstruction process on the CelebA dataset under varying Gaussian noise magnitudes $\sigma^2$ added to CNN gradients. It reveals that the PSNR between the intermediate reconstructions $\hat{\mathbf{x}}_0(\mathbf{x}_t)$ and the original target image $\mathbf{x}_{\text{private}}$ exhibits a non-monotonic behavior, peaking before subsequently declining. This phenomenon is attributed to the inherent denoising capability of diffusion models. Rather than directly converging from a randomly initialized dummy image to the distorted target $\mathbf{x}_{\text{noisy}}$, the reconstruction process produces smoother intermediate results that are more accurate approximations of $\mathbf{x}_{\text{private}}$. However, as reconstruction iterations progress, diffusion model’s denoising capacity becomes insufficient to counteract the noise inherent in $\mathbf{x}_{\text{noisy}}$. Consequently, in later stages, the reconstruction converges toward $\mathbf{x}_{\text{noisy}}$ and deviates from $\mathbf{x}_{\text{private}}$, which explains the decline in GIA performance. The quantitative results in Table \ref{reconstruction_noisy_table} corroborate these observations.

\begin{figure}[ht]
  \centering
    \includegraphics[width=0.33\textwidth]{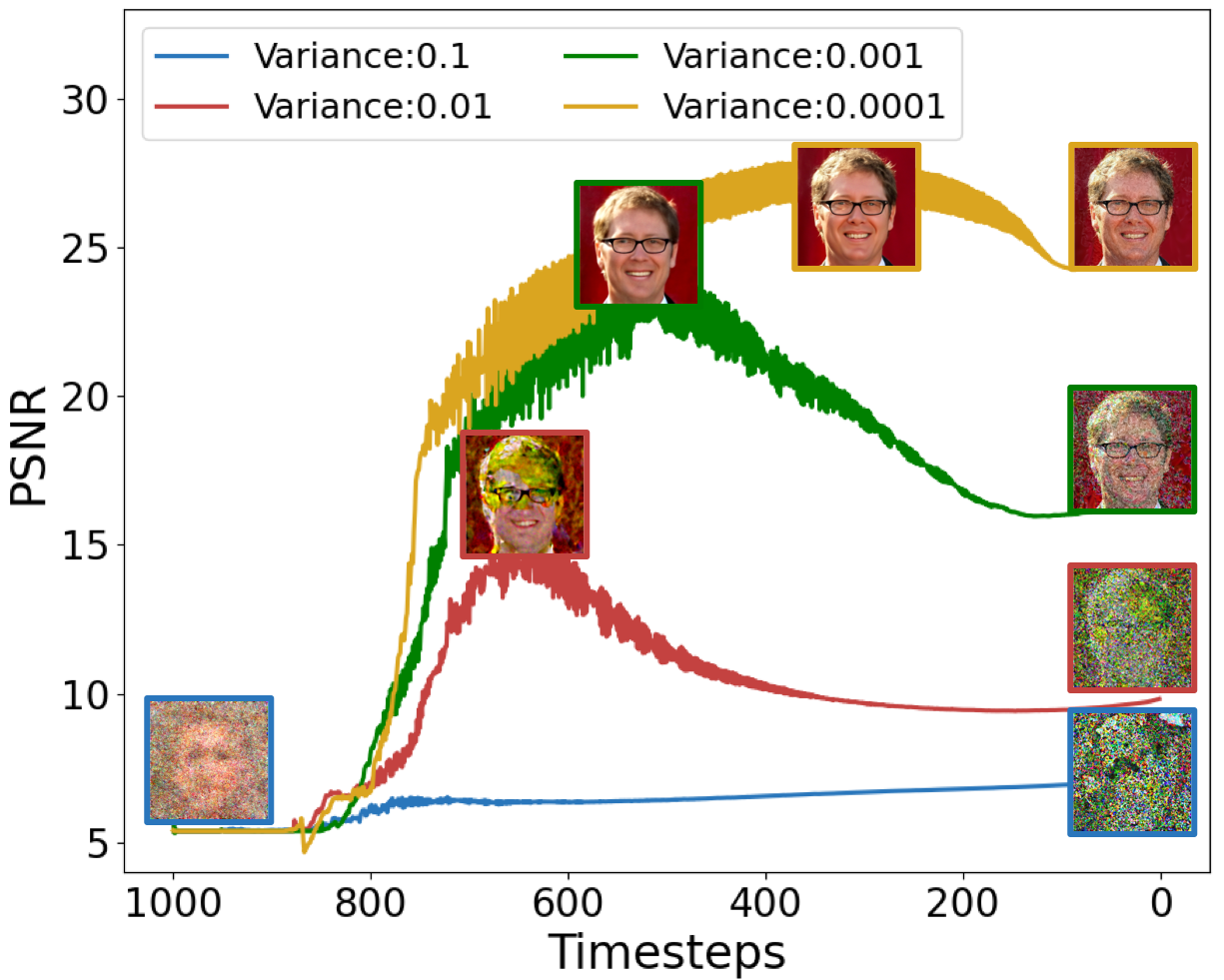}
  \caption{Reconstruction processes on CelebA with Gaussian noise-perturbed CNN gradients. Reconstructions at peak performance and upon GIA completion are displayed.}
  \label{noisy_celeba_process}
\end{figure}

\begin{table}[!htb]
    \centering
    \resizebox{0.4\textwidth}{!}{
    \begin{tabular}{cccc}
    \hline
    \toprule
      Gaussian Variance&  MSE$\downarrow$  & PSNR$\uparrow$ & LPIPS$\downarrow$\\
    \midrule
     $10^{-4}$ & 0.0016  & 28.0343 & 7.9935e-6\\
            $10^{-3}$ & 0.0039 & 24.0567  &  4.5263e-5\\
             $10^{-2}$ & 0.0366 & 14.9784  &  2.4762e-4 \\
             $10^{-1}$ & 0.1252  & 9.0236 & 7.1237e-4\\ 
    \bottomrule
    \end{tabular}}
    \caption{Quantitative analysis on reconstruction quality on CelebA with Gaussian noise-perturbed CNN gradients. Metrics are reported between the intermediate reconstructions $\hat{\mathbf{x}}_0(\mathbf{x}_t)$ at peak performance and the original target $\textbf{x}_{\text{private}}$.}
    \label{reconstruction_noisy_table}
\end{table}

Experimental results further corroborate our theoretical findings: (i) Effect of noise scale on reconstruction error. Reconstruction quality degrades noticeably with increasing noise levels. This observation aligns with our theoretical analysis (Theorems \ref{UpperBoundJensen}-\ref{lowerboundDSG}), which establishes that larger noise levels make accurate reconstruction provably more challenging; (ii) Reduction of attack loss during the reconstruction process. According to Theorem \ref{upperDSG}, the attack loss in Algorithm \ref{alg:1} decreases monotonically as the process evolves from random noise towards the original target $\textbf{x}_{\text{private}}$, achieving optimal reconstruction upon convergence. In our implementation, line 7 of Algorithm \ref{alg:1} is modified to Equation (\ref{equation_dsg}), leveraging the DDIM generative prior to improve reconstruction smoothness and quality. This modification yields earlier PSNR peaks, though reconstruction quality eventually declines due to the noise-induced gradient perturbations; (iii) Convergence rate of attack loss. Theorem \ref{DSGlower} states that the attack loss reduction rate is governed by the noise scale: higher noise levels yield smaller decreases $\mathcal{L}(\mathbf{x}_{t-1})-\mathcal{L}(\mathbf{x}_t)$, leading to slower convergence and poorer final reconstructions. This trend is empirically validated in Figure \ref{noisy_celeba_process}, which exhibits a clear correlation between increased noise variance and both slower PSNR improvement and reduced reconstruction fidelity.

\subsubsection{Comparison with Baselines.}

To demonstrate the superior performance of GGSS-R, we conducted a comprehensive comparison with state-of-the-art baselines on CelebA, using CNN as the attacked model. Table \ref{reconstruction_baseline_table} reports the quantitative results of GIAs from the original gradients. As shown in Table \ref{reconstruction_baseline_table}, our approach successfully guides the diffusion model to reconstruct high-resolution images with markedly improved fidelity, underscoring the heightened privacy risks when adversaries gain access to the original gradients.

\begin{table}[!htb]
    \centering
    \resizebox{0.47\textwidth}{!}{
    \begin{tabular}{lccc}
        \toprule
        Method & MSE$\downarrow$& PSNR$\uparrow$ & LPIPS$\downarrow$ \\
        \midrule
        DLG \cite{b5} & 0.0480& 13.1835 & 2.3453e-4\\
        IG \cite{b7} & 0.0196 & 17.0756 & 1.1612e-4 \\
        GI \cite{b8} & 0.0223  & 16.5109 & 1.4274e-4\\
        GIAS \cite{b10} & 0.0458  & 13.3885 & 3.9351e-4 \\
        GGL \cite{b9}& 0.0179& 17.4923 & 1.1937e-4 \\
        GIFD \cite{b11}& 0.0098 & 20.0534 &6.6672e-5 \\
        FinetunedDiff \cite{b25}& 0.0011 & 29.4490 & 5.4441e-6 \\
        GGSS-R (Ours)& \textbf{0.0001} & \textbf{41.1229} & \textbf{2.0519e-7} \\
        \bottomrule
    \end{tabular}}
    \caption{Quantitative analysis of baselines and our method with original, unperturbed gradients.}
    \label{reconstruction_baseline_table}
\end{table}

We further extended the comparison to assess reconstruction performance using CNN gradients perturbed by Gaussian noise ($\sigma^2 = 0.01$). As detailed in Table \ref{noisy_baseline_table}, our methods demonstrate improved robustness against noise-perturbed gradients. Specifically, GGSS-R achieves the highest PSNR, surpassing the optimal baseline, GGL, by $1.0802$ dB. These results underscore that simple noise addition is insufficient to ensure gradient privacy, highlighting the need for more sophisticated mitigation strategies in practical applications.

\begin{table}[!htb]
    \centering
    \resizebox{0.47\textwidth}{!}{
    \begin{tabular}{lccc}
        \toprule
        Method & MSE$\downarrow$ & PSNR$\uparrow$ & LPIPS$\downarrow$\\
        \midrule
        DLG \cite{b5} & 0.2194& 6.5871 & 1.8540e-3\\
        IG \cite{b7} & 0.0747 & 11.2677 & 3.2103e-4 \\
        GI \cite{b8} & 0.0892 & 10.4968 & 5.6760e-4\\
        GIAS \cite{b10} & 0.0613 & 12.1276 & 3.4494e-4 \\
        GGL \cite{b9}& 0.0394& 13.8982 & 2.5247e-4 \\
        GIFD \cite{b11}& 0.0425 &  13.7118 &2.5901e-4 \\
        FinetunedDiff \cite{b25}& 0.0683 & 11.6675 & 3.0877e-4 \\
        GGSS-R (Ours)& \textbf{0.0366} & \textbf{14.9784} & \textbf{2.4762e-4}\\
        \bottomrule
    \end{tabular}}
    \caption{Quantitative analysis of baselines and our method with Gaussian noise ($\sigma^2=0.01$)-perturbed gradients.}
    \label{noisy_baseline_table}
\end{table}

\subsection{Ablation Studies}

\subsubsection{Impact of Batch Size.}

We analyze the effect of batch size $\mathcal{B}$ for gradient acquisition on GIA performance, using the CelebA dataset. As illustrated in Table \ref{batchsize}, reconstruction quality consistently degrades as the batch size increases across all GIA methods. This deterioration stems from the averaging effect inherent to batching, which obscures image-specific gradient information and consequently reduces the reconstruction fidelity. Nevertheless, GGSS-R consistently outperforms all baselines at every batch size, benefiting from the strong generative capacity of diffusion models coupled with pixel-level conditional guidance.

\begin{table}[!htb]
    \centering
    \resizebox{0.48\textwidth}{!}{
    \begin{tabular}{lcccc}
        \toprule
        Method & $\mathcal{B}=1$& $\mathcal{B}=2$ & $\mathcal{B}=4$ & $\mathcal{B}=8$ \\
        \midrule
        DLG \cite{b5} & 13.1835& 11.0356 &9.9452 &8.3562 \\
        IG \cite{b7} & 17.0756& 15.2562 &11.3003 &9.5382 \\
        GI \cite{b8}  &16.5109 & 14.8351 &11.0352 &9.2350\\
        GIAS \cite{b10}  & 13.3885& 11.4684 &10.2465 &8.9462\\
        GGL \cite{b9} & 17.4923& 15.5526 & - & -\\
        GIFD \cite{b11} &20.0534 & 18.3562 &13.9636 & 10.3425\\
        FinetunedDiff \cite{b25} &29.4490 & 21.4583 & 15.6387& 11.4570\\
        GGSS-R (Ours) &\textbf{41.1229} & \textbf{24.4602} &\textbf{17.4791} & \textbf{12.1425}\\
        \bottomrule
    \end{tabular}}
    \caption{Impact of batch size for gradient acquisition on reconstruction quality. Note that StyleGAN2's latent vector contains a large number of parameters and the CMA-ES optimizer used in GGL does not support large-scale optimization. As a result, GGL is limited to batch sizes $\mathcal{B} \le 2$.
    }
    \label{batchsize}
\end{table}

\subsubsection{Impact of Guidance Rate.}

Table \ref{scale_DSG} reports the effect of the guidance rate $m_r$ on reconstruction quality. The guidance rate, explored over the range $0.01$–$0.80$, governs the balance between conditional guidance and generated image quality. The optimal GIA performance is achieved at $m_r = 0.20$, yielding the lowest MSE ($7.7215e-5$), the highest PSNR ($41.1230$) and the lowest LPIPS ($2.0519e-7$). 

\begin{table}[!htp]
    \centering
    \resizebox{0.38\textwidth}{!}{
    \begin{tabular}{cccc}
    \hline
    \toprule
        Guidance Rate & MSE$\downarrow$ & PSNR$\uparrow$ & LPIPS$\downarrow$ \\
    \midrule
             0.01 & 2.0406e-3  & 26.9023 & 6.9703e-6\\
             0.10 & 2.1778e-4  & 36.6198 & 4.1730e-7\\
             0.20 & \textbf{7.7215e-5}  & \textbf{41.1229} & \textbf{2.0519e-7}\\
             0.30 & 1.4769e-4  & 38.3065 & 2.8713e-7\\
             0.40 & 1.0778e-4 & 39.6746 & 1.9574e-7\\
             0.60 & 2.0407e-4  & 36.9023 & 9.0661e-7\\
             0.80 & 2.7478e-4 & 35.6101 & 1.1251e-6\\
    \bottomrule
    \end{tabular}}
    \caption{Impact of guidance rate on reconstruction quality.}
    \label{scale_DSG}
\end{table}

\subsubsection{Impact of Noise Type.}

To evaluate how noise type influences reconstruction performance in noise-based defenses, we varied the magnitude of Laplacian noise applied to gradients. Table \ref{reconstruction_lap_table} quantitatively evaluates the corresponding peak reconstruction performance. Comparing these results with those obtained under Gaussian noise perturbations (Table \ref{reconstruction_noisy_table}), we observe that the reconstruction quality is determined primarily by the noise scale, while the specific noise distribution has only a negligible effect.

\begin{table}[ht]
    \centering
    \resizebox{0.38\textwidth}{!}{
    \begin{tabular}{cccc}
    \hline
    \toprule
     Laplacian Variance &  MSE$\downarrow$ &PSNR$\uparrow$ &LPIPS$\downarrow$ \\
    \midrule
    
      $10^{-4}$ & 0.0016   & 27.9395 & 9.3180e-6\\
         $10^{-3}$ & 0.0038 & 24.2073 & 3.7156e-5 \\
         $10^{-2}$ & 0.0304 & 15.1657 & 2.3746e-4\\
        $10^{-1}$ & 0.1390 & 8.5697 & 5.6100e-4\\
    \bottomrule
    \end{tabular}}
    \caption{Quantitative analysis on reconstruction quality with Laplacian noise-perturbed CNN gradients.}
    \label{reconstruction_lap_table}
\end{table}

More details of running time and GPU usage are provided in Appendix C.

\section{Conclusion and Future Work}

This paper proposes a novel GIA method based on conditional diffusion models to recover high-quality images from noise-perturbed gradients. Our approach requires minimal modifications to the reverse diffusion sampling and operates without prior knowledge. By exploiting the inherent generative capability of diffusion models, our attack effectively mitigates the impact of noise perturbation. Moreover, we derive theoretical bounds on the reconstruction error and analyze the convergence properties of the attack loss. Our findings further reveal an intrinsic vulnerability of diverse model architectures to GIAs, and we propose the RV metric to quantify this vulnerability. Extensive experiments confirm the superior performance of our method in recovering high-quality reconstructions from noise-perturbed gradients.

\section*{Acknowledgments}
This work is supported by National Natural Science Foundation
of China (U24B20144, U23A20299, 62172424, 62276270, 62322214,
62436010, 62441230), The Natural Science Foundation of Fujian Province (2024J08276, 2024J08278), The Key Research Project for Young and Middle-aged Researchers by the Fujian Provincial Department of Education (JZ230044).

\bibliography{aaai2026}

\appendix
\onecolumn 
\section{Proofs of Theorems}

\subsection{Proof of Theorem 5.1}
\textbf{Theorem 5.1 (Upper Bound of the Reconstruction Error).}
Under the attacked model $F(\mathbf{x}; W)$, we assume that the adversary obtains a Gaussian noise–perturbed gradient $\mathbf{g}_{\text{leaked}} = \nabla_{W} F(\mathbf{x}_\text{private};W) + \mathcal{N}(0, \sigma^2 I)$. The gradient inversion attack follows Algorithm~\ref{alg:1}, using the approximation $p({\mathbf{g}_{\text{leaked}}} | \mathbf{x}_{t}) \simeq p({\mathbf{g}_{\text{leaked}}} | \hat{\mathbf{x}}_{0}(\mathbf{x}_{t}))$, where $\hat{\mathbf{x}}_{0}(\mathbf{x}_{t})$ denotes the posterior mean of $\mathbf{x}_0$ given $\mathbf{x}_t$. The reconstruction error—measured by the Jensen gap—is upper bounded by:
        \begin{equation*}
        \begin{split}
            &\mathcal{J}\left( \mathbf{g}(\mathbf{x}), p(\mathbf{x}_0 | \mathbf{x}_{t}) \right)=|\mathbb{E}[\mathbf{g}(\mathbf{x}_0)]-\mathbf{g}(\mathbb{E}[\mathbf{x}_0])| \\
            & \leq \frac{n}{\sqrt{2\pi \sigma^{2}}} \| \nabla_{\mathbf{x}} \mathbf{g}(\mathbf{x}) \| \int \| \mathbf{x}_0 - \hat{\mathbf{x}}_{0}(\mathbf{x}_{t}) \|p(\mathbf{x}_0 | \mathbf{x}_{t}) d\mathbf{x}_0,
        \end{split}
        \end{equation*}
        where $n$ is $\mathbf{x}$'s dimension, and $\mathbf{g}(\mathbf{x})=\nabla_{\mathbf{x}} \nabla_{W} F(\mathbf{x}; W)$.

\begin{proof}[Proof of Theorem 5.1]
Analogous to the proof of Theorem 1 in \cite{b19}, the upper bound of the Jensen Gap is:
\begin{equation}
\begin{split}
    &\mathcal{J}\left( \mathbf{g}(\mathbf{x}), p(\mathbf{x}_0 | \mathbf{x}_{t}) \right) \\
    & \leq \frac{n}{\sqrt{2\pi \sigma^{2}}} \exp(- \frac{1}{2\sigma^2}) \| \nabla_{\mathbf{x}}\mathbf{g}(\mathbf{x})\| \int \| \mathbf{x}_0 - \hat{\mathbf{x}}_{0}(\mathbf{x}_{t}) \| p(\mathbf{x}_0 | \mathbf{x}_{t})d\mathbf{x}_0\\
    & \leq \frac{n}{\sqrt{2\pi \sigma^{2}}} \| \nabla_{\mathbf{x}} \mathbf{g}(\mathbf{x})\| \int \| \mathbf{x}_0 - \hat{\mathbf{x}}_{0}(\mathbf{x}_{t}) \| p(\mathbf{x}_0 | \mathbf{x}_{t})d\mathbf{x}_0.
\end{split}
\end{equation}
The final inequality holds because $\frac{\exp(-\frac{1}{\sigma^2})}{\sigma^2} \leq \frac{1}{\sigma^2}$.
\end{proof}

\subsection{Proof of Theorem 5.2}

\textbf{Theorem 5.2 (Lower Bound of the Reconstruction Error).}
Under conditions of Theorem \ref{UpperBoundJensen} and Assumptions \ref{assu1}-\ref{assu4}, for $\mathbf{x} \in \mathbb{R}^n$, the reconstruction error is lower bounded by:
    {
    \begin{equation*}
        \begin{split}
        & \mathcal{J}\left( \mathcal{L}\left( \mathbf{g}(\mathbf{x}),\mathbf{g}_{\text{leaked}} \right), p(\mathbf{x}_0 | \mathbf{x}_{t}) \right)\\
        &=|\mathbb{E}[\mathcal{L}(\mathbf{g}(\mathbf{x}_0),\mathbf{g}_{\text{leaked}})]-\mathcal{L}(\mathbf{g}(\mathbb{E}[\mathbf{x}_0]),\mathbf{g}_{\text{leaked}})| \\
        & \geq \frac{1}{2} \left( \beta \lambda_{min}\left( \mathbf{J}_\mathbf{g}(\mathbf{x})^T\mathbf{J}_\mathbf{g}(\mathbf{x}) \right) - \alpha L_g \right) \sum_{i=1}^{n} \sigma_{i}^{2},
        \end{split}
    \end{equation*}
    }
    where $\lambda_{\text{min}}(\cdot)$ represents the smallest eigenvalue of $\cdot$, $L_g$ can be refined to a specific value $\frac{n}{\sqrt{2 \pi \sigma^2}} \exp(-\frac{1}{2\sigma^2})$ under Gaussian noise-perturbed isotropic gradients, and $\{ \sigma_i^2 \}_{i=1}^n$ are the eigenvalues of the covariance matrix of $p(\mathbf{x}_0 | \mathbf{x}_{t})$.

To prove Theorem 5.2, we introduce Lemma \ref{lemma2} and Lemma \ref{lemma3}.

\begin{lemma}[\textbf{\cite{b19}}]
    Under the conditions of Theorem 5.1, if the attack loss $\mathcal{L}$ is $\beta^{\prime}$-strongly convex with respect to $\mathbf{x}$ and $p(\mathbf{x}_0| \mathbf{x}_{t})$ is a Gaussian distribution with covariance matrix whose eigenvalues are $\{ \sigma_i^2 \}_{i=1}^n$. Then the lower bound of Jensen Gap is bounded by:
\begin{equation}
    \mathcal{J}\left( \mathcal{L}\left( \mathbf{g}(\mathbf{x})),\mathbf{g}_{\text{leaked}} \right), p(\mathbf{x}_0 | \mathbf{x}_{t}) \right) \geq \frac{1}{2} \beta^{\prime} \sum_{i=1}^{n} \sigma_{i}^{2}.
\end{equation}
\label{lemma2}
\end{lemma}

\begin{lemma}[\textbf{\cite{b19}}]
    Let $\phi(\cdot)$ be an isotropic multivariate Gaussian density function with variance matrix $\sigma^2 I$. There exists a constant $L$ such that $\forall \mathbf{x}, \mathbf{y} \in \mathbb{R}^{n}$,
    \begin{equation}
        \| \phi(\mathbf{x}) - \phi(\mathbf{y}) \| \leq L \| \mathbf{x} - \mathbf{y} \|,
    \end{equation}
    where $L = \frac{n}{\sqrt{2 \pi \sigma^2}} \exp(-\frac{1}{2\sigma^2})$.
    \label{lemma3}
\end{lemma}
We omit the proofs of Lemmas \ref{lemma2} and \ref{lemma3} as they are provided in \cite{b19}.

\begin{proof}[Proof of Theorem 5.2]

For simplicity, we denote $\mathcal{L}(\mathbf{g}(\mathbf{x}),\mathbf{g}_{\text{leaked}})$ as $\mathcal{L}(\mathbf{g}(\mathbf{x}))$ within the proof. Taking Taylor's expansion of $\mathcal{L}(\mathbf{g}(\mathbf{x}))$ around $\mathbf{g}(\mathbf{y})$, we get:
\begin{equation}
        \mathcal{L}(\mathbf{g}(\mathbf{x})) = \mathcal{L}(\mathbf{g}(\mathbf{y}))  + \left( \mathbf{J}_\mathbf{g}(\mathbf{x})^T \nabla_\mathbf{g} \mathcal{L}(\mathbf{g}(\mathbf{x})) \right)^T (\mathbf{x} - \mathbf{y}) 
        + \frac{1}{2} (\mathbf{x} - \mathbf{y})^T \nabla^2 \mathcal{L}(\mathbf{g}(\mathbf{c})) (\mathbf{x} - \mathbf{y}),
    \end{equation}
where $\nabla^2 \mathcal{L}(\mathbf{g}(\mathbf{x}))$ denotes the Hessian matrix of $\mathcal{L}$ with respect to $\mathbf{x}$, $\mathbf{J}_\mathbf{g}(\mathbf{x})= \nabla_\mathbf{x}\mathbf{g}(\mathbf{x})$ is the Jacobian of $\mathbf{g}(\mathbf{x})$, and $\mathbf{c}$ is a point such that $\mathbf{g}(\mathbf{c}) \in [\mathbf{g}(\mathbf{x}), \mathbf{g}(\mathbf{y})]$.

Next, we analyze the Hessian matrix $\nabla^2 \mathcal{L}(\mathbf{g}(\mathbf{x}))$:
\begin{equation}
            \nabla^2 \mathcal{L}(\mathbf{g}(\mathbf{x})) = \nabla_{\mathbf{x}}\left( \mathbf{J}_\mathbf{g}(\mathbf{x})^T \nabla_\mathbf{g} \mathcal{L}(\mathbf{g}(\mathbf{x})) \right) = \mathbf{J}_\mathbf{g}(\mathbf{x})^T \nabla_\mathbf{g}^{2} \mathcal{L}(\mathbf{g}(\mathbf{x})) \mathbf{J}_\mathbf{g}(\mathbf{x}) + \sum_{i=1}^{m} \nabla \mathcal{L}_i(\mathbf{g}(\mathbf{x})) \nabla_{\mathbf{x}}^2 \mathbf{g}_i(\mathbf{x}),
    \end{equation}
where $\nabla \mathcal{L}_i(\mathbf{g}(\mathbf{x})) = \frac{\partial \mathcal{L}}{\partial \mathbf{g}_i(\mathbf{x})}$ and $m$ is the dimension of $\mathbf{g}(\mathbf{x})$. $\nabla^2 \mathbf{g}_i(\mathbf{x})$ is the Hessian matrix of the $i$-th component of $\mathbf{g}(\mathbf{x})$ with respect to $\mathbf{x}$. Since $\mathcal{L}$ is $\beta$-convex with respect to $\mathbf{g}(\mathbf{x})$, we have $\beta I \preceq \nabla_\mathbf{g}^2 \mathcal{L}(\mathbf{g}(\mathbf{x}))$. This implies:
\begin{equation}
        \beta \lambda_{\text{min}}\left( \mathbf{J}_\mathbf{g}(\mathbf{x})^T\mathbf{J}_\mathbf{g}(\mathbf{x}) \right) I  \preceq \beta  \mathbf{J}_\mathbf{g}(\mathbf{x})^T\mathbf{J}_\mathbf{g}(\mathbf{x})\preceq \mathbf{J}_\mathbf{g}(\mathbf{x})^T \nabla_\mathbf{g}^{2} \mathcal{L}(\mathbf{g}(\mathbf{x})) \mathbf{J}_\mathbf{g}(\mathbf{x}),
    \end{equation}
where $\lambda_{\text{min}}(\cdot)$ represents the smallest eigenvalue of a matrix.

Furthermore, given that $\mathbf{g}(\mathbf{x})$ is $L_g$-smooth with respect to $\mathbf{x}$, we have $\| \nabla^2 \mathbf{g}_i(\mathbf{x}) \| \leq L_g$. And since $\mathcal{L}$ is $\alpha$-smooth with respect to $\mathbf{g}(\mathbf{x})$, it follows that $\| \nabla \mathcal{L}_i(\mathbf{g}(\mathbf{x})) \| \leq \alpha$. Therefore, for the summation term:
\begin{equation}
        \| \sum_{i=1}^{m} \nabla \mathcal{L}_i(\mathbf{g}(\mathbf{x})) \nabla_{\mathbf{x}}^2 \mathbf{g}_i(\mathbf{x}) \| \leq m \alpha L_g.
    \end{equation}
    Combining these results, we establish a lower bound for $\nabla^2 \mathcal{L}(\mathbf{g}(\mathbf{x}))$:
    \begin{equation}
        \beta \lambda_{\text{min}}\left( \mathbf{J}_\mathbf{g}(\mathbf{x})^T\mathbf{J}_\mathbf{g}(\mathbf{x}) \right) I - m \alpha L_g I \preceq \nabla^2 \mathcal{L}(\mathbf{g}(\mathbf{x}))
    \end{equation}
    The supremum of the smallest eigenvalue of $\nabla^2 \mathcal{L}(\mathbf{g}(\mathbf{c}))$ is therefore:
    \begin{equation}
        \sup (\lambda_{\text{min}}(\nabla^2 \mathcal{L}(\mathbf{g}(\mathbf{c})))) = \beta \lambda_{\text{min}}\left( \mathbf{J}_\mathbf{g}(\mathbf{x})^T\mathbf{J}_\mathbf{g}(\mathbf{x}) \right) - \alpha L_g.
    \end{equation}
    Consequently, $\mathcal{L}(\mathbf{g}(\mathbf{x}))$ is $\beta^{\prime}$-convex, where $\beta^{\prime}= \beta \lambda_{\text{min}}\left( \mathbf{J}_\mathbf{g}(\mathbf{x})^T\mathbf{J}_\mathbf{g}(\mathbf{x}) \right) - \alpha L_g$. According to Lemma \ref{lemma2}, the lower bound of the Jensen Gap is bounded by $\frac{1}{2} \beta^{\prime} \sum_{i=1}^{n} \sigma_{i}^{2}$:
\begin{equation}
    \mathcal{J}\left( \mathcal{L}\left( \mathbf{g}(\mathbf{x}) \right), p(\mathbf{x}_0 | \mathbf{x}_{t}) \right) =\mathcal{J}\left( \mathcal{L}\left( \mathbf{g}(\mathbf{x}),\mathbf{g}_{\text{leaked}} \right), p(\mathbf{x}_0 | \mathbf{x}_{t}) \right) \geq \frac{1}{2} \left( \beta \lambda_{\text{min}}\left( \mathbf{J}_\mathbf{g}(\mathbf{x})^T\mathbf{J}_\mathbf{g}(\mathbf{x}) \right) - \alpha L_g \right) \sum_{i=1}^{n} \sigma_{i}^{2}.
    \end{equation}
    For the case in which $\mathbf{g}(\mathbf{x})$ is a gradient perturbed by Gaussian noise, Lemma \ref{lemma3} refines $L_g$ to a specific value $L_g = \frac{n}{\sqrt{2 \pi \sigma^2}} \exp(-\frac{1}{2\sigma^2})$.
\end{proof}

\subsection{Proof of Theorem 5.3}
\textbf{Theorem 5.3 (Convergence of the Attack Loss).}
Assume that the attack loss $\mathcal{L}(\mathbf{x})$ is $\alpha^{\prime}$-smooth, $\beta^{\prime}$-strongly convex with respect to $\mathbf{x}$. When GIA is executed following Algorithm~\ref{alg:1}, the sequence $\{ \mathcal{L}(\mathbf{x}_{t}) \}_{t=T}^{1}$ monotonically decreases as the denoising steps progress from timestep $t=T$ to $1$.

To prove Theorem 5.3, Lemma \ref{lemma4} is necessary. We omit its proof for brevity, as it is straightforward to derive.

\begin{lemma}
    Assume that the attack loss $\mathcal{L}(\mathbf{x})$ is $\alpha^{\prime}$-smooth and $\beta^{\prime}$-strongly convex with respect to $\mathbf{x}$. Let $H(\mathcal{L}(\mathbf{x}))$ represents the Hessian of $\mathcal{L}$ with respect to $\mathbf{x}$, we have: 
    \begin{equation}
        \beta^{\prime} I \preceq H(\mathcal{L}(\mathbf{x})) \preceq \alpha^{\prime} I,
    \end{equation}
    \label{lemma4}
where $I$ is the identity matrix.
\end{lemma}

\begin{proof}[Proof of Theorem 5.3]

Take Taylor expansion on $\mathcal{L}(\mathbf{x})$ around $\mathbf{x}_{t}$, which states:
\begin{equation}
        \mathcal{L}(\mathbf{x}) = \mathcal{L}(\mathbf{x}_{t}) + \left ( \nabla_{\mathbf{x}} \mathcal{L}(\mathbf{x}_t) \right)^T (\mathbf{x} - \mathbf{x}_{t}) + \frac{1}{2} (\mathbf{x} - \mathbf{x}_{t})^T H(\mathcal{L}(\mathbf{c})) (\mathbf{x} - \mathbf{x}_{t}),
    \end{equation}
where $\mathbf{c}$ is some point between $\mathbf{x}$ and $\mathbf{x}_{t}$.

Subsequently, we have:
    \begin{equation}
    \begin{split}
         \mathcal{L}(\mathbf{x}_{t-1}) -  \mathcal{L}(\mathbf{x}_{t}) & = \left ( \nabla_{\mathbf{x}} \mathcal{L}(\mathbf{x}_t) \right)^T (\mathbf{x}_{t-1} - \mathbf{x}_{t}) + \frac{1}{2} (\mathbf{x}_{t-1} - \mathbf{x}_{t})^T H(\mathcal{L}(\mathbf{c})) (\mathbf{x}_{t-1} - \mathbf{x}_{t})\\
        & \leq \left ( \nabla_{\mathbf{x}} \mathcal{L}(\mathbf{x}_t) \right)^T (\mathbf{x}_{t-1} - \mathbf{x}_{t}) + \frac{\alpha^{\prime}}{2} (\mathbf{x}_{t-1} - \mathbf{x}_{t})^T (\mathbf{x}_{t-1} - \mathbf{x}_{t}).
    \end{split}
\label{DSG-upp1}
\end{equation}
The inequality holds because the eigenvalues of $H(\mathcal{L}(\mathbf{c}))$ are bounded within $[ \beta^{\prime}, \alpha^{\prime} ]$. Since $\mathbf{x}_{t-1} = \arg \min_{\mathbf{x}} \left ( \nabla_{\mathbf{x}} \mathcal{L}(\mathbf{x}_t) \right)^T (\mathbf{x} - \mathbf{x}_{t})$ where $\mathbf{x} \in S_{\mu_\theta(\mathbf{x}_t,t),\sqrt{n}\sigma_t}^{n}$, substituting any other value for $\mathbf{x}_{t-1}$ would amplify the upper bound in Equation (\ref{DSG-upp1}). By replacing $\mathbf{x}_{t-1}$ with $\mathbf{x}_{t-1}^{+}$ where 
\begin{equation}
        \mathbf{x}_{t-1}^{+} = \frac{\left ( \nabla_{\mathbf{x}} \mathcal{L}(\mathbf{x}_t) \right)^T \mathbf{x}_{t} - \frac{\alpha^{\prime}}{2}(\mathbf{x}_{t-1} - \mathbf{x}_{t})^T (\mathbf{x}_{t-1} - \mathbf{x}_{t})}{\left ( \nabla_{\mathbf{x}} \mathcal{L}(\mathbf{x}_t) \right)^T \left ( \nabla_{\mathbf{x}} \mathcal{L}(\mathbf{x}_t) \right)} \cdot  \nabla_{\mathbf{x}} \mathcal{L}(\mathbf{x}_t).
\end{equation}
we have:
\begin{equation}
        \mathcal{L}(\mathbf{x}_{t-1}) -  \mathcal{L}(\mathbf{x}_{t}) \leq \left ( \nabla_{\mathbf{x}} \mathcal{L}(\mathbf{x}_t) \right)^T (\mathbf{x}_{t-1}^{+} - \mathbf{x}_{t}) + \frac{\alpha^{\prime}}{2} (\mathbf{x}_{t-1} - \mathbf{x}_{t})^T (\mathbf{x}_{t-1} - \mathbf{x}_{t}) = 0.
\end{equation}  
\end{proof}

\subsection{Proof of Theorem 5.4}

\textbf{Theorem 5.4 (Upper Bound on the Convergence Rate of the Attack Loss).}
Under Assumptions \ref{assu1}-\ref{assu4}, and assume that the attack loss $\mathcal{L}(\mathbf{x})$ is $\alpha^{\prime}$-smooth, $\beta^{\prime}$-strongly convex with respect to $\mathbf{x}$. The optimization using Algorithm \ref{alg:1} satisfies:
        \begin{equation*}
        \begin{split}
            \mathcal{L}(\mathbf{x}_{t-1}) - \mathcal{L}(\mathbf{x}_{t})& \geq  \frac{\beta^{\prime} n \sigma_t^2}{2} +\beta^{\prime} k \sqrt{n} \sigma_t + \frac{\beta^{\prime} k^2}{2}\\
            -  ( k  + &  \sqrt{n} \sigma_t) (\sqrt{\lambda_{\text{max}}(\mathbf{J}_\mathbf{g}(\mathbf{x})\mathbf{J}_\mathbf{g}(\mathbf{x})^T)} \cdot L_g),
        \end{split}
        \end{equation*}
    where $n$ is $\mathbf{x}$'s dimension, $k$ is a positive constant, $\lambda_{\text{max}}(\cdot)$ is the largest eigenvalue of $\cdot$, and $L_g = \frac{n}{\sqrt{2 \pi \sigma^2}} \exp(-\frac{1}{2\sigma^2})$.

The derivation of the lower bound for the attack loss decrease rate in Theorem 5.4 leverages Lemma \ref{lemma5}, which establishes an upper bound for $\| \nabla_{\mathbf{x}} \mathcal{L}(\mathbf{x}_t) \|_2$. 
\begin{lemma}[Upper bound of $\| \nabla_{\mathbf{x}} \mathcal{L}(\mathbf{x}_t) \|_2$]
    Under conditions in Assumptions 5.1-5.3, the upper bound of $\| \nabla_{\mathbf{x}} \mathcal{L}(\mathbf{x}_t) \|_2$ is:
    \begin{equation}
        \| \nabla_{\mathbf{x}} \mathcal{L}(\mathbf{x}_t) \|_2 \leq \sqrt{\lambda_{\text{max}}(\mathbf{J}_\mathbf{g}(\mathbf{x})\mathbf{J}_\mathbf{g}(\mathbf{x})^T)} \cdot L_g,
    \end{equation}
   where $\mathbf{x}\in \mathbb{R}^{n}$, and $L_g = \frac{n}{\sqrt{2 \pi \sigma^2}} \exp(-\frac{1}{2\sigma^2})$.
   \label{lemma5}
\end{lemma}

We begin by proving Lemma \ref{lemma5}.

\begin{proof}[Proof of Lemma A.4]
Expanding $\| \nabla_{\mathbf{x}} \mathcal{L}(\mathbf{x}) \|_2^2 $:
\begin{equation}
        \begin{split}
            \| \nabla_{\mathbf{x}} \mathcal{L}(\mathbf{x}) \|_2^2 &= \left( \nabla_{\mathbf{x}} \mathcal{L}(\mathbf{x}) \right)^T \left( \nabla_{\mathbf{x}} \mathcal{L}(\mathbf{x}) \right)\\
            &= \left( \mathbf{J}_\mathbf{g}(\mathbf{x})^T \nabla_\mathbf{g} \mathcal{L}(\mathbf{g}(\mathbf{x})) \right)^T \left( \mathbf{J}_\mathbf{g}(\mathbf{x})^T \nabla_\mathbf{g} \mathcal{L}(\mathbf{g}(\mathbf{x})) \right)\\
            &= \nabla_\mathbf{g} \mathcal{L}(\mathbf{g}(\mathbf{x}))^T \left( \mathbf{J}_\mathbf{g}(\mathbf{x})\mathbf{J}_\mathbf{g}(\mathbf{x})^T \right)\nabla_\mathbf{g} \mathcal{L}(\mathbf{g}(\mathbf{x})).
        \end{split}
    \end{equation}
Assuming that the eigendecomposition of $\mathbf{J}_\mathbf{g}(\mathbf{x})\mathbf{J}_\mathbf{g}(\mathbf{x})^T$ is $\mathbf{J}_\mathbf{g}(\mathbf{x})\mathbf{J}_\mathbf{g}(\mathbf{x})^T = U^T \Lambda U$, then we have:
\begin{equation}
        \begin{split}
            \| \nabla_{\mathbf{x}} \mathcal{L}(\mathbf{x}) \|_2^2 &= \nabla_\mathbf{g} \mathcal{L}(\mathbf{g}(\mathbf{x}))^T \left( U^T \Lambda U \right)\nabla_\mathbf{g} \mathcal{L}(\mathbf{g}(\mathbf{x})) \\
            &= \nabla_\mathbf{g} \mathcal{L}(\mathbf{g}(\mathbf{x}))^T \left( U^T \Lambda^{\frac{1}{2}} \Lambda^{\frac{1}{2}} U\right)\nabla_\mathbf{g} \mathcal{L}(\mathbf{g}(\mathbf{x}))\\
            &=\nabla_\mathbf{g} \mathcal{L}(\mathbf{g}(\mathbf{x}))^T \left( U^T (\Lambda^{\frac{1}{2}})^T \Lambda^{\frac{1}{2}} U\right)\nabla_\mathbf{g} \mathcal{L}(\mathbf{g}(\mathbf{x}))\\
            &= \left( \Lambda^{\frac{1}{2}} U \nabla_\mathbf{g} \mathcal{L}(\mathbf{g}(\mathbf{x})) \right)^T \left( \Lambda^{\frac{1}{2}} U \nabla_\mathbf{g} \mathcal{L}(\mathbf{g}(\mathbf{x})) \right)\\
            &= \| \Lambda^{\frac{1}{2}} U \nabla_\mathbf{g} \mathcal{L}(\mathbf{g}(\mathbf{x})) \|_2^2.
        \end{split}
    \end{equation}
Thus, 
    \begin{equation}
        \begin{split}
            \| \nabla_{\mathbf{x}} \mathcal{L}(\mathbf{x}) \|_2 &= \| \Lambda^{\frac{1}{2}} U \nabla_\mathbf{g} \mathcal{L}(\mathbf{g}(\mathbf{x})) \|_2\\
            & \leq \| \Lambda^{\frac{1}{2}} \|_2 \| U \nabla_\mathbf{g} \mathcal{L}(\mathbf{g}(\mathbf{x})) \|_2 = \| \Lambda^{\frac{1}{2}} \|_2 \| \nabla_\mathbf{g} \mathcal{L}(\mathbf{g}(\mathbf{x})) \|_2 \\
            & \leq \sqrt{\lambda_{\text{max}}(\mathbf{J}_\mathbf{g}(\mathbf{x})\mathbf{J}_\mathbf{g}(\mathbf{x})^T)} \cdot L_g,
        \end{split}
    \end{equation}
    where $L_g = \frac{n}{\sqrt{2 \pi \sigma^2}} \exp(-\frac{1}{2\sigma^2})$.
\end{proof}

We then prove Theorem 5.4.

\begin{proof}[Proof of Theorem 5.4]
According to Lemma \ref{lemma4}:
\begin{equation}
    \begin{split}
        \mathcal{L}(\mathbf{x}_{t-1}) -  \mathcal{L}(\mathbf{x}_{t}) & = \left ( \nabla_{\mathbf{x}} \mathcal{L}(\mathbf{x}_t) \right)^T (\mathbf{x}_{t-1} - \mathbf{x}_{t}) + \frac{1}{2} (\mathbf{x}_{t-1} - \mathbf{x}_{t})^T H(\mathcal{L}(\mathbf{c})) (\mathbf{x}_{t-1} - \mathbf{x}_{t})\\
        & \geq \underset{(*)}{\underbrace{\left ( \nabla_{\mathbf{x}} \mathcal{L}(\mathbf{x}_t) \right)^T (\mathbf{x}_{t-1} - \mathbf{x}_{t})}} + \underset{(**)}{\underbrace{\frac{\beta^{\prime}}{2} (\mathbf{x}_{t-1} - \mathbf{x}_{t})^T (\mathbf{x}_{t-1} - \mathbf{x}_{t})}}
    \end{split}
    \label{DSG-low1}
    \end{equation}

In Equation (\ref{DSG-low1}), $\mathbf{x}_{t-1}$ is defined as $\arg \min_{\mathbf{x}} \left ( \nabla_{\mathbf{x}} \mathcal{L}(\mathbf{x}_t) \right)^T (\mathbf{x} - \mathbf{x}_{t})$ for $\mathbf{x} \in S_{\mu_{\theta}(\mathbf{x}_t,t), \sqrt{n} \sigma_t}^{n}$. In our specific case, we obtain $\mathbf{x}_{t-1} = \mu_{\theta}(\mathbf{x}_{t},t) - \sqrt{n} \sigma_t d^*$ where $d^* = \frac{\nabla_{\mathbf{x}} \mathcal{L}(\mathbf{x}_t)}{\| \nabla_{\mathbf{x}} \mathcal{L}(\mathbf{x}_t) \|_2}$. Consequently, for $(*)$, we have: 

\begin{equation}
\begin{split}
    (*) &= \left ( \nabla_{\mathbf{x}} \mathcal{L}(\mathbf{x}_t) \right)^T (\mathbf{x}_{t-1}  - \mu_{\theta}(\mathbf{x}_{t},t) + \mu_{\theta}(\mathbf{x}_{t},t) - \mathbf{x}_{t})\\
    &= -\sqrt{n} \sigma_t \| \nabla_{\mathbf{x}} \mathcal{L}(\mathbf{x}_t) \|_2 + \left ( \nabla_{\mathbf{x}} \mathcal{L}(\mathbf{x}_t) \right)^T \left(\mu_{\theta}(\mathbf{x}_{t},t) - \mathbf{x}_{t}\right).
\end{split}
\end{equation}
For $(**)$, we have:
\begin{equation}
    \begin{split}
        (**) &= \frac{\beta^{\prime}}{2} (\mathbf{x}_{t-1} - \mathbf{x}_{t})^T (\mathbf{x}_{t-1} - \mathbf{x}_{t}) \\
        & = \frac{\beta^{\prime}}{2} (\mathbf{x}_{t-1} - \mu_{\theta}(\mathbf{x}_{t},t))^T (\mathbf{x}_{t-1} - \mu_{\theta}(\mathbf{x}_{t},t)) \\
        & \quad +\frac{\beta^{\prime}}{2} (\mu_{\theta}(\mathbf{x}_{t},t)-\mathbf{x}_{t})^T (\mu_{\theta}(\mathbf{x}_{t},t)-\mathbf{x}_{t}) \\
        & \quad + \mu (\mathbf{x}_{t-1} - \mu_{\theta}(\mathbf{x}_{t},t))^T \left(\mu_{\theta}(\mathbf{x}_{t},t) - \mathbf{x}_{t}\right)\\
        &=\frac{\beta^{\prime} n \sigma_t^2}{2} + \frac{\beta^{\prime}}{2} \left ( 2\mathbf{x}_{t-1} - \mathbf{x}_{t} - \mu_{\theta}(\mathbf{x}_{t},t)\right)^T\left(\mu_{\theta}(\mathbf{x}_{t},t) - \mathbf{x}_{t}\right)\\
        &= \frac{\beta^{\prime} n \sigma_t^2}{2} + \frac{\beta^{\prime}}{2} \left ( \mathbf{x}_{t-1} - \mathbf{x}_{t} \right)^T\left(\mu_{\theta}(\mathbf{x}_{t},t) - \mathbf{x}_{t}\right) \\
        & \quad + \frac{\beta^{\prime}}{2} \left ( \mathbf{x}_{t-1} - \mu_{\theta}(\mathbf{x}_{t},t)\right)^T\left(\mu_{\theta}(\mathbf{x}_{t},t) - \mathbf{x}_{t}\right).
    \end{split}
    \label{(**)}
\end{equation}
To derive the lower bound of $\mathcal{L}(\mathbf{x}_{t-1}) -  \mathcal{L}(\mathbf{x}_{t})$, the common factor identified is $\mu_{\theta}(\mathbf{x}_{t},t) - \mathbf{x}_{t}$. The terms $\frac{\beta^{\prime}}{2} \left ( \mathbf{x}_{t-1} - \mathbf{x}_{t} \right)^T\left(\mu_{\theta}(\mathbf{x}_{t},t) - \mathbf{x}_{t}\right)$ and $\frac{\beta^{\prime}}{2} \left ( \mathbf{x}_{t-1} - \mu_{\theta}(\mathbf{x}_{t},t)\right)^T\left(\mu_{\theta}(\mathbf{x}_{t},t) - \mathbf{x}_{t}\right)$ in Equation (\ref{(**)}) represent two inner products. To demonstrate how the denoising ability can reduce the attack loss, we seek a value that is demonstrably achievable. Considering the randomness in $\mu_{\theta}(\mathbf{x}_{t},t)$, by setting $\mu_{\theta}(\mathbf{x}_{t},t) - \mathbf{x}_{t} = -kd^* (k>0)$, the term $\frac{\beta^{\prime}}{2} \left ( \mathbf{x}_{t-1} - \mu_{\theta}(\mathbf{x}_{t},t)\right)^T\left(\mu_{\theta}(\mathbf{x}_{t},t) - \mathbf{x}_{t}\right)$ in Equation (\ref{(**)}) can be expressed as:
\begin{equation}
        \quad \frac{\beta^{\prime}}{2} \left ( \mathbf{x}_{t-1} - \mu_{\theta}(\mathbf{x}_{t},t)\right)^T\left(\mu_{\theta}(\mathbf{x}_{t},t) - \mathbf{x}_{t}\right) = - \frac{\beta^{\prime}}{2} (\sqrt{n} \sigma_t d^*)^T (-kd^*)
        = \frac{\beta^{\prime} k \sqrt{n} \sigma_t }{2}.
    \end{equation}
Then, the term $\frac{\beta^{\prime}}{2} \left ( \mathbf{x}_{t-1} - \mathbf{x}_{t} \right)^T\left(\mu_{\theta}(\mathbf{x}_{t},t) - \mathbf{x}_{t}\right)$ in Equation (\ref{(**)}) is:
    \begin{equation}
        \begin{split}
            &\frac{\beta^{\prime}}{2} \left ( \mathbf{x}_{t-1} - \mathbf{x}_{t} \right)^T\left(\mu_{\theta}(\mathbf{x}_{t},t) - \mathbf{x}_{t}\right) \\=& \frac{\beta^{\prime}}{2} \left ( \mathbf{x}_{t-1} - \mu_{\theta}(\mathbf{x}_{t},t) -kd^* \right)^T\left(-kd^*\right)\\ 
            =& - \frac{\beta^{\prime}}{2} \left( (\sqrt{n} \sigma_t - k) d^* \right)^T(-kd^*)\\
            =& \frac{\beta^{\prime} k(\sqrt{n} \sigma_t + k)}{2}. 
        \end{split}
    \end{equation}
The $(*)$ is:
\begin{equation}
\begin{split}
    (*) & = -\sqrt{n} \sigma_t \| \nabla_{\mathbf{x}} \mathcal{L}(\mathbf{x}_t) \|_2 + \left ( \nabla_{\mathbf{x}} \mathcal{L}(\mathbf{x}_t) \right)^T \left(\mu_{\theta}(\mathbf{x}_{t},t) - \mathbf{x}_{t}\right)\\
    &= -\sqrt{n} \sigma_t \| \nabla_{\mathbf{x}} \mathcal{L}(\mathbf{x}_t) \|_2 - k \| \nabla_{\mathbf{x}} \mathcal{L}(\mathbf{x}_t) \|_2.
\end{split}
\end{equation}
Finally, we have:
\begin{equation}
    \begin{split}
        & \mathcal{L}(\mathbf{x}_{t-1}) -  \mathcal{L}(\mathbf{x}_{t}) \\
        & \geq \frac{\beta^{\prime} n \sigma_t^2}{2} +\beta^{\prime} k \sqrt{n} \sigma_t - (k + \sqrt{n} \sigma_t) \| \nabla_{\mathbf{x}} \mathcal{L}(\mathbf{x}_t) \|_2 + \frac{\beta^{\prime} k^2}{2}\\
        & \geq \frac{\beta^{\prime} n \sigma_t^2}{2} +\beta^{\prime} k \sqrt{n} \sigma_t - (k + \sqrt{n} \sigma_t) (\sqrt{\lambda_{\textbf{max}}(\mathbf{J}_\textbf{g}(\mathbf{x})\mathbf{J}_\textbf{g}(\mathbf{x})^T)} \cdot L_g) + \frac{\beta^{\prime} k^2}{2}.
    \end{split}
    \label{DSG-low2}
    \end{equation}
\end{proof}

To better understand the proof, see Figure \ref{noise_lap_appendix} for details. If $\mu_{\theta}(\mathbf{x}_{t},t)- \mathbf{x}_{t} = -kd^* (k>0)$ is not satisfied, it is the case in Figure \ref{DSG_01}. Due to the randomness of $\mathbf{x}_{t-1}$, $\mu_{\theta}(\mathbf{x}_{t},t) - \mathbf{x}_{t} = -kd^* (k>0)$ can be satisfied which is the case in Figure \ref{DSG_02}. In this case, $\mathcal{L}(\mathbf{x}_{t-1}) -  \mathcal{L}(\mathbf{x}_{t})$ reaches its lower bound.

\begin{figure}[htp!]
    \centering
    \vspace{-2mm}
    \subfigure[$\mu_{\theta}(\mathbf{x}_{t},t) - \mathbf{x}_{t} \neq -kd^*$]{
		\includegraphics[width=0.4\textwidth]{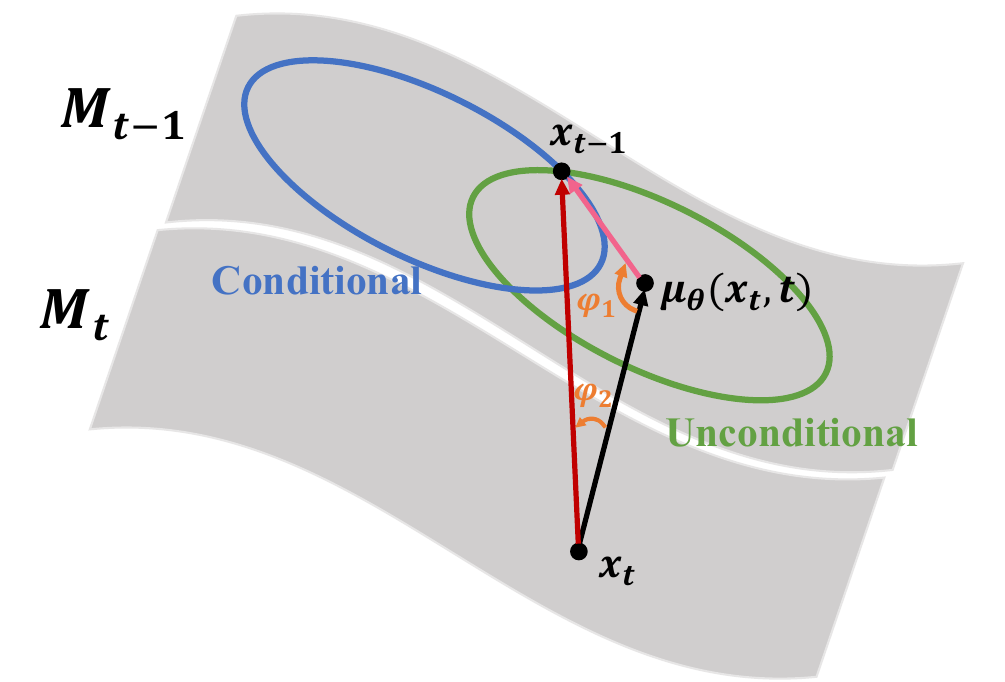}
		\label{DSG_01}
    }
    \subfigure[$\mu_{\theta}(\mathbf{x}_{t},t) - \mathbf{x}_{t} = -kd^*$]{
		\includegraphics[width=0.4\textwidth]{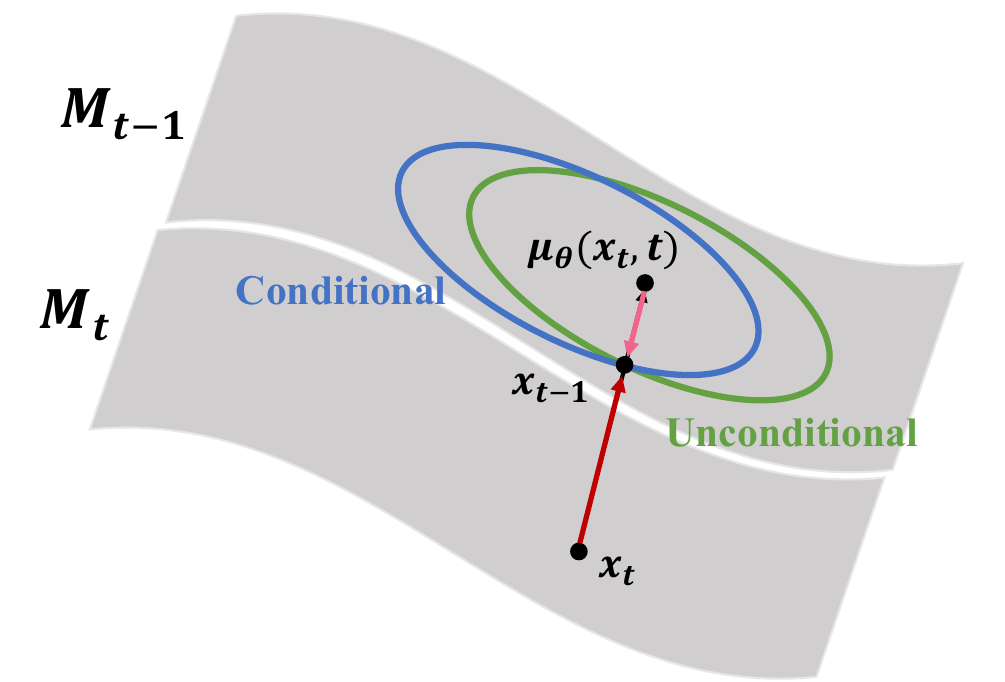}
		\label{DSG_02}
    }
    \caption{Explanations for Theorem 5.4.}
    \label{noise_lap_appendix}
\end{figure}

\section{Additional Experimental Settings}

\subsection{Machine Configuration}

All experiments were performed on a server equipped with dual Intel(R) Xeon(R) Silver 4310 CPUs @2.10GHz and a single A100 GPU. We implement our attack with Pytorch 2.0.1.

\subsection{Architectures of Attacked Models}

The architectures of attacked models are detailed as follows:

$\bullet$ \textit{MLP-3/4/5.} These three Multilayer Perceptrons process flattened image inputs and utilize ReLU activations in their hidden layers. MLP-3 is a three-layer network with two hidden layers (512 and 128 units). MLP-4 extends this with three hidden layers (512, 256, and 64 units). MLP-5 further deepens the architecture, featuring four hidden layers (1028, 512, 256, and 64 units). All MLP models conclude with a linear output layer for classification.

$\bullet$ \textit{CNN.} The CNN architecture consists of two convolutional layers (64 and 128 output channels, respectively, with $3\times3$ kernels and padding), each followed by ReLU activation and $2 \times2$ max-pooling. Feature maps are then flattened and processed by a 256-unit ReLU fully connected layer, culminating in a softmax output for classification. This CNN serves as the default attacked model in our experiments.

$\bullet$ \textit{ResNet18.} ResNet18 employs residual blocks with skip connections. It begins with an $7\times7$ convolutional layer and max-pooling, followed by four stages, each containing two residual blocks. Each block utilizes two $3\times3$ convolutional layers, batch normalization, and ReLU activation. The network concludes with adaptive average pooling and a fully connected output layer for classification.

$\bullet$ \textit{VGGNet16.} Our VGGNet16 implementation adheres to the classic architecture, featuring five convolutional blocks. The initial two blocks each contain two $3\times3$ convolutional layers, and the subsequent blocks employ three $3\times3$ convolutional layers. All convolutional layers are followed by ReLU activations and $2\times2$ max-pooling. Flattened features then pass through a fully connected classifier with two 4096-unit ReLU layers (with dropout) and a final softmax output.

\section{Running Time and GPU Usage}

Table \ref{running_time} provides an overview of the computational resources consumed during the reconstruction process. Specifically, it details the average running time (in seconds) per conditional reverse step for GGSS-R when applied to various attacked model architectures. Furthermore, the table presents the peak GPU usage (in MiB) observed throughout the reconstruction process.

\begin{table}[!htp]
    \centering
    \begin{tabular}{l c c} 
    \toprule
        Attacked Model &Running Time (s) / step & Peak GPU Usage (MiB) \\ 
    \midrule
            ResNet18   & 0.2259  & 3396 \\
            MLP-5 & 2.6473  &9680  \\
            VGGNet16  & 6.3582  & 30314 \\
            CNN &  2.3947 &  7822\\
            MLP-4 &  1.9149 &  6084 \\
            MLP-3 &  1.4939 &  5247\\
    \bottomrule
    \end{tabular}
    \caption{Average running time (s) per gradient-guided reverse diffusion sampling step and peak GPU usage (MiB) for GGSS-R.}
    \label{running_time}
\end{table}

\end{document}